\documentclass[11pt,a4paper]{article}

\usepackage[margin=1.1in]{geometry}
\usepackage[utf8]{inputenc}
\usepackage[T1]{fontenc}
\usepackage{mathptmx}
\usepackage{amsmath,amssymb,amsthm,mathtools}
\usepackage{booktabs}
\usepackage{microtype}
\usepackage{xcolor}
\usepackage{algorithm}
\usepackage{algpseudocode}
\usepackage{enumitem}
\usepackage[numbers,compress,sort]{natbib}
\usepackage{hyperref}

\allowdisplaybreaks

\newtheorem{theorem}{Theorem}[section]
\newtheorem{lemma}[theorem]{Lemma}
\newtheorem{corollary}[theorem]{Corollary}
\newtheorem{proposition}[theorem]{Proposition}
\theoremstyle{definition}
\newtheorem{definition}[theorem]{Definition}
\newtheorem{problem}[theorem]{Problem}
\theoremstyle{remark}
\newtheorem{remark}[theorem]{Remark}

\newcommand{\E}{\mathbb{E}}
\newcommand{\R}{\mathbb{R}}
\newcommand{\F}{\mathcal{F}}
\newcommand{\OPT}{\mathrm{OPT}}
\newcommand{\Tr}{\mathrm{Tr}}

\newcommand{\eps}{\varepsilon}
\newcommand{\poly}{\mathrm{poly}}
\newcommand{\supp}{\mathrm{supp}}
\newcommand{\Topk}{\mathrm{Top}_k}
\newcommand{\Sg}{\mathrm{Sg}}
\DeclareMathOperator{\Bern}{Bern}

\title{A PTAS for Non-Adaptive Stochastic Top-$k$ Sum\\
under General Combinatorial Constraints\thanks{A large language model was used to assist with drafting and copy-editing. The author verified all scientific claims, proofs, and references.}}

\author{%
  Yu Liu \\
  \texttt{liuyujyyz@gmail.com}
}

\begin{document}
\maketitle

\begin{abstract}
We study non-adaptive selection of a feasible set $S$ so as to maximize the expected sum of the $k$ largest realized values among independent nonnegative discrete random variables.
The same objective arises when hiring a team of $k$ workers or when computing VCG welfare in an $\ell$-unit auction.
The main setting is a fixed-dimensional nonnegative packing family: the natural LP has $d=O(1)$ packing inequalities with binary coefficients.
No single algorithm achieves a constant factor on every membership family (already at $k=1$).
Given an $\alpha$-approximate max-sum oracle, a decreasing surplus search yields ratio $\alpha/((1+\alpha)(1+\eps))$ for every $k\ge 1$ (cuts included).
Every fixed-$d$ packing family already has a deterministic max-sum PTAS, hence inherits that constant.
Unlike schemes for $\E[\max]$, which encode a set by a void probability, the decoder uses two additive approximations to the occupancy functional $p\mapsto\E[\min(k,N(p))]$: type-histograms when $k=O(1/\eps^2)$, and a three-dimensional mixture-quantile type when $k=\Omega(1/\eps^2)$.
Those signatures are realized by a packing LP rather than by exact-sum, after enumerating $n^{f(d,1/\eps)}$ heavy items.
The result is a PTAS for every $k\ge 1$ on every fixed-$d$ packing family, including binary one- and two-dimensional knapsack.
In this packing setting the scheme is essentially optimal as a generic guarantee: there is no FPTAS that works for every such $\F$ unless $P=NP$, and no EPTAS unless $W[1]=FPT$ (two-dimensional knapsack is a witness, already at $k=1$).
A separate boundary is query-weight exact-sum, which includes DAG paths and matchings and is incomparable with fixed-$d$ packing.
That oracle also yields a PTAS for every $k$, so $d$-dimensional packing is a useful taxonomy, not a partition of every family that admits a PTAS.
The minimization problem $\min_{S\in\F}\E[\Topk(S)]$ uses the same signatures: an upward min-sum search yields a constant factor, query-weight exact-sum again yields a PTAS, and every fixed-$d$ covering family admits a PTAS by realizing those signatures in the covering LP (fractionals are raised to $1$).
Two-dimensional covering knapsack rules out a generic FPTAS on that class.
\end{abstract}

\section{Introduction}
\label{sec:intro}

A recruiter must choose a feasible team before any interview outcome is observed.
Each candidate $i$ has an independent, known discrete law $X_i$ for the value they will deliver, and the team's output is the sum of the $k$ strongest realized performances: weaker members are slack, not a linear bonus.
The same functional appears as VCG welfare in an $\ell$-unit auction when the platform first shortlists a feasible set of bidders and only then runs the auction~\citep{kleinberg2018,psomas2020,gravin2024}.
We study this problem in the fully non-adaptive model: the algorithm returns $S\in\F\subseteq 2^{[n]}$ with no access to realizations.
Write $\Topk(S)$ for the sum of the $k$ largest coordinates of $(X_i)_{i\in S}$ (pad with zeros if $|S|<k$) and $\OPT=\max_{S\in\F}\E[\Topk(S)]$.

When $k=1$ the objective is the expected maximum, or maximum-element.
That problem has been studied in depth.
It is already NP-hard under cardinality~\citep{goel2010,psomas2020}.
It admits a PTAS under a query-weight exact-sum oracle~\citep{chen2016,liu2017}, and an EPTAS under cardinality~\citep{psomas2020,segev2021,segev2024}.

Several of those conclusions have been pushed from a single maximum toward Top-$k$, but only under a cardinality (or $\ell$-unit) constraint.
Cardinality-constrained welfare in position auctions is a linear combination of order statistics and now has a PTAS~\citep{gravin2024}.
\citet{segev2021} sketch an extension of their ProbeMax scheme from one maximum to $\ell$-unit welfare.
Adaptive probing yields a PTAS for the expected sum of the $k$ largest probed values when $k=O(1)$ and at most $m$ variables may be opened sequentially~\citep{fu2018}.
The $k$-th order statistic itself (MAX-$k$, not the Top-$k$ sum) is already hard for $k=2$ under cardinality~\citep{psomas2020}.

Under a general family $\F$, those conclusions do not extend to $\E[\Topk]$ in any natural way.
A PTAS for the maximum is not a PTAS for the sum of the $k$ largest: $\Topk$ is not a maximum over an expanded ground set.
The existing $k=1$ algorithms match a void probability $\prod_i(1-q_i)$~\citep{chen2016} or a single truncated-mean scale~\citep{liu2017}.
Arbitrary $k$ changes the stochastic object.
By layer cake the objective is an integral of the truncated occupancy functional $p\mapsto\E[\min(k,N(p))]$, and matching one void probability does not control that functional.
Naive truncation of large atoms likewise destroys lotteries that carry $\Theta(\OPT)$ of the objective.
The $k=1$ analyses therefore do not yield a PTAS for $k\neq 1$.
Cardinality schemes~\citep{segev2021,gravin2024} and the adaptive constant-$k$ PTAS~\citep{fu2018} remain tied to a cardinality (or probe-budget) constraint: they do not give a decoder for a general packing or covering family, and they do not treat growing $k$ in the non-adaptive packing model.

The algorithmic difference is the sketch.
We keep the surplus scale and the mean-preserving fold from the $k=1$ theory, and we replace the void-probability signature by two additive approximations of $p\mapsto\E[\min(k,N(p))]$: an occupancy histogram with Poissonization when $k=O(1/\eps^2)$, and a three-dimensional surplus type at a mixture quantile when $k=\Omega(1/\eps^2)$.
Either sketch interfaces with combinatorial feasibility in the same way a one-dimensional type would: by query-weight exact-sum, or by a basic feasible solution of a fixed-dimensional packing or covering LP after enumerating $n^{f(d,1/\eps)}$ heavy items.
The two ranges meet at $k=\Theta(1/\eps^2)$.

The main line of the paper is a four-rung ladder on a single class: fixed-dimensional nonnegative packing (Definition~\ref{def:packing}).
If max-sum on an arbitrary $\F$ admits an $\alpha$-approximation, a decreasing surplus search yields ratio $\alpha/((1+\alpha)(1+\eps))$ for every $k\ge 1$, including cuts via Goemans--Williamson~\citep{goemans1995}.
Every fixed-$d$ packing family already has a deterministic max-sum PTAS~\citep{frieze1984,ibarra1975}, so it inherits that constant automatically.
On packing the two occupancy sketches are realized by the natural LP (Section~\ref{sec:packing}), so every fixed-$d$ packing family admits a PTAS for every $k\ge 1$, including binary one- and two-dimensional knapsack.

In this packing setting the generic guarantee is essentially optimal.
Already at $k=1$, sparse lotteries encode deterministic max-sum into $\E[\Topk]$.
There is no single constant-factor algorithm on every membership family~\citep{zuckerman2006}, and no single PTAS on every family that already has constant-factor max-sum~\citep{papadimitriou1991,hastad2001}.
On packing itself there is no FPTAS that works for every $\F$ unless $P=NP$, and no EPTAS unless $W[1]=FPT$: two-dimensional knapsack lies in the class and supplies the witness~\citep{magazine1984,kulik2010}.
A finer scheme on a concrete packing family must use extra structure of that $\F$; cardinality for $k=1$ already has an EPTAS.

Fixed-$d$ packing is a classification, not a partition of every family that admits a PTAS.
Query-weight exact-sum (Definition~\ref{def:exact-sum}) is an incomparable sufficient condition: DAG paths and matchings have it and are not fixed-$d$ packing; binary knapsack is packing and does not have it.
Section~\ref{sec:exactsum} records that the same signatures, realized by exact-sum rather than by the packing LP, yield a PTAS on every query-weight family.
The two boundaries meet on cardinality and on unary knapsack; neither contains the other.

Minimization of the same functional is the covering counterpart.
An upward min-sum search yields a constant factor (Theorem~\ref{thm:min-transfer}), as for $\E[\max]$ in~\citet{liu2017}.
Every fixed-$d$ covering family already has a deterministic min-sum PTAS (the covering analogue of~\citet{frieze1984}), hence inherits that constant.
The same two occupancy sketches, realized in the covering LP by raising at most $O(d)$ fractionals to $1$, give a PTAS for every $k\ge 1$ (Theorem~\ref{thm:covering}), including binary covering knapsack.
Minimization is therefore the covering counterpart of the same decoder, not a second algorithm.
Two-dimensional covering knapsack rules out a generic FPTAS on that class (Theorem~\ref{thm:min-limits}).
Query-weight exact-sum is again a second boundary (Corollary~\ref{cor:min-allk}): $s$--$t$ paths, spanning trees, and assignment are neither packing nor covering.

\paragraph{Contributions.}
\begin{enumerate}[leftmargin=1.4em,itemsep=2pt]
\item There is no single algorithm that achieves a constant factor on every membership family $\F$ (already at $k=1$). Given an $\alpha$-approximate max-sum oracle, a decreasing surplus search yields ratio $\frac{\alpha}{(1+\alpha)(1+\eps)}$ for every $k\ge 1$ (Theorem~\ref{thm:transfer}). Every fixed-$d$ packing family has a max-sum PTAS, hence sits on this rung.
\item There is no single PTAS on every family that already has constant-factor max-sum (already at $k=1$; Theorem~\ref{thm:noptas}). On every fixed-$d$ packing family the problem admits a PTAS for every $k\ge 1$, by realizing two approximations of $p\mapsto\E[\min(k,N(p))]$---occupancy histograms or a mixture-quantile type---in the packing LP (Theorem~\ref{thm:packing}). Those signatures are the difference from the $k=1$ void-probability schemes.
\item There is no FPTAS that works for every fixed-$d$ packing family unless $P=NP$, and no EPTAS unless $W[1]=FPT$ (Theorem~\ref{thm:limits}; two-dimensional knapsack). In the packing setting a generic PTAS is therefore essentially optimal.
\item Query-weight exact-sum is a second, incomparable PTAS boundary (Theorems~\ref{thm:main} and~\ref{thm:quantile}, Corollary~\ref{cor:allk}, Section~\ref{sec:exactsum}).
\item The same ladder holds for $\min_{S\in\F}\E[\Topk(S)]$: a min-sum transfer (Theorem~\ref{thm:min-transfer}), a PTAS on every fixed-$d$ covering family (Theorem~\ref{thm:covering}), a matching FPTAS lower bound (Theorem~\ref{thm:min-limits}), and a query-weight PTAS (Corollary~\ref{cor:min-allk}).
\end{enumerate}
Proofs are in the appendix.

\section{Related work}
\label{sec:related}

\paragraph{Expected maxima.}
Selecting a feasible set to maximize $\E[\max_{i\in S}X_i]$ is the $k=1$ case of our problem.
\citet{goel2010} and later \citet{psomas2020} established NP-hardness already under cardinality.
\citet{chen2016} gave the first PTAS for general $\F$ via Bernoulli decomposition of the maximum and signatures that preserve the void probability.
The truncation scale, the geometric max-sum search, and the mean-preserving fold for $\E[\max]$ were developed in~\citet{liu2017}.
Those arguments are for $k=1$: the signature is a void probability, which determines $\E[\min(1,N)]$ and nothing more.
They do not yield a PTAS for $\E[\Topk]$ when $k\neq 1$, and the occupancy and mixture-quantile sketches below are not a reuse of that scalar at a larger $k$.
\citet{psomas2020} and \citet{segev2021,segev2024} obtained EPTASes under cardinality (non-adaptive ProbeMax), which are stronger than Theorem~\ref{thm:packing} when $k=1$ and $\F$ is a cardinality constraint.

\paragraph{Order statistics, teams, and auctions.}
\citet{kleinberg2018} study team performance with test scores.
\citet{psomas2020} distinguish the expected maximum from the $k$-th order statistic and prove strong hardness for the latter.
VCG welfare in an $\ell$-unit auction is exactly $\E[\Topk]$; GSP equilibria match VCG welfare in position auctions~\citep{edelman2007,varian2007}.
\citet{gravin2024} give a Poisson relaxation and, with brute force for small $k$, a PTAS for cardinality-constrained position-auction welfare, together with a sketched $\ell$-unit extension of~\citet{segev2021}.

\paragraph{Adaptive probing and stochastic DP.}
In the adaptive ProbeMax / ProbeTop-$k$ model the algorithm opens variables sequentially and may stop~\citep{gupta2016probing,fu2018}.
\citet{fu2018} give a PTAS for the adaptive expected Top-$k$ \emph{sum} under a cardinality probe budget, for constant $k$.
That scheme does not apply to a general family $\F$, does not produce a non-adaptive decoder, and does not approximate $p\mapsto\E[\min(k,N(p))]$ for growing $k$.

\paragraph{Minimizing $\E[\Topk]$.}
The opposite problem $\min_{S\in\F}\E[\Topk(S)]$ is the stochastic Top-$k$ norm (a $k$-sum / $k$-centrum criterion).
\citet{ibrahimpur2020} give approximation algorithms for stochastic min-norm load balancing and spanning trees, by reducing a symmetric norm to a small collection of expected Top-$k$ norms.
Deterministic $k$-sum is classical in location and robust routing: one minimizes the $k$ largest distances or the $k$ largest edge delays on a path or tree.
Those applications live on covering families and on bases / paths, not on downward-closed packing.
An upward min-sum search supplies a constant-factor scale $W$ (Theorem~\ref{thm:min-transfer}).
Fixed-$d$ covering realizes the signatures in the covering LP (Theorem~\ref{thm:covering}); query-weight exact-sum is a second PTAS boundary (Corollary~\ref{cor:min-allk}).

\paragraph{Stochastic combinatorial optimization.}
Thresholded surplus, Poisson-binomial tails, and Le Cam approximation are standard tools in stochastic packing and expected-utility maximization~\citep{li2011,lecam1960}.
Submodular maximization~\citep{nemhauser1978,sviridenko2004,calinescu2011} already gives a $(1-1/e)$-approximation because $S\mapsto\E[\Topk(S)]$ is monotone submodular (Appendix~\ref{app:submod}); a PTAS is a different guarantee.

Table~\ref{tab:landscape} records the comparison.
The main guarantee is a PTAS for $\E[\Topk]$ on every fixed-$d$ packing family, for every $k\ge 1$ (Theorem~\ref{thm:packing}), above a constant-factor scheme on every family with approximate max-sum (Theorem~\ref{thm:transfer}).
On packing there is no generic FPTAS or EPTAS (Theorem~\ref{thm:limits}).
Query-weight exact-sum is a second PTAS boundary (Corollary~\ref{cor:allk}).
The minimization counterparts are Theorems~\ref{thm:min-transfer},~\ref{thm:covering} and~\ref{thm:min-limits}, and Corollary~\ref{cor:min-allk}.

\begin{table}[t]
\caption{Where the main packing ladder and the exact-sum boundary sit.
``Card.'' means a cardinality (or $\ell$-unit shortlist) constraint.
Fu et al.\ is adaptive probing, not the non-adaptive model.}
\label{tab:landscape}
\centering
\small
\begin{tabular}{@{}llll@{}}
\toprule
Paper & Objective & Feasibility & Guarantee \\
\midrule
\citet{chen2016} & $\E[\max]$ & exact-sum $\F$ & PTAS \\
\citet{psomas2020} & $\E[\max]$ / $k$-th OS & card. & EPTAS / hard \\
\citet{segev2021} & ProbeMax & card.\ (sketched $\ell$-unit) & EPTAS \\
\citet{gravin2024} & position welfare & card. & PTAS \\
\citet{fu2018} & adaptive $\E[\Topk]$, $k=O(1)$ & card.\ probe budget & PTAS \\
this paper & $\E[\Topk]$ (any $k$) & max-sum $\F$ & const.\ (not uniform) \\
this paper & $\E[\Topk]$ (any $k$) & fixed-$d$ packing & PTAS; no uniform FPTAS/EPTAS \\
this paper & $\E[\Topk]$ (any $k$) & query-wt.\ exact-sum & PTAS \\
this paper & $\min\E[\Topk]$ (any $k$) & min-sum $\F$ & const.\ (not uniform) \\
this paper & $\min\E[\Topk]$ (any $k$) & fixed-$d$ covering & PTAS; no uniform FPTAS \\
this paper & $\min\E[\Topk]$ (any $k$) & query-wt.\ exact-sum & PTAS \\
\bottomrule
\end{tabular}
\end{table}

\section{Technical overview}
\label{sec:overview}

Both sketches share a surplus scale and a mean-preserving fold with the $\E[\max]$ theory; they differ from that theory in the additive signature, and they differ from each other in whether the signature is an occupancy histogram or a mixture-quantile triple.
Realization is then by a packing LP or by exact-sum.

\paragraph{A scale from truncated means.}
For a threshold $T$, the surplus of a set is $\sum_{i\in S}\E[(X_i-T)_+]$.
There is always a $T$ at which the maximum surplus over $\F$ equals $kT$.
Pathwise, $\Topk(x)\le kT+\sum_i(x_i-T)_+$, so $\OPT\le 2kT$ with no independence.
Independence is used only on the matching lower bound: at the surplus-maximizing set, layer cake plus the Poisson-binomial inequality $\E[\min(k,N)]\ge k\mu/(k+\mu)$ produce $\E[\Topk]\ge kT$.
Thus $W:=kT$ satisfies $W\le\OPT\le 2W$, and every later error budget is a fraction of $W$.
The $k=1$ case is the truncation argument of~\citet{liu2017}.
If max-sum is available only approximately, the same comparison is run as a decreasing search on $T$: halt at the first height where the oracle set has surplus greater than $kT$.
The output is at least $kT_0$ and every feasible set is at most $kT_0(1+\eps)(1+1/\alpha)$, which is ratio $\alpha/((1+\alpha)(1+\eps))$ independently of $k$.

\paragraph{Fold, then grid.}
Atoms larger than $M=W/\eps$ cannot be discarded.
A single lottery $B(H,p)$ with mean $\Theta(W)$ and $H\gg M$ has $\E[\Topk]=\Theta(W)$, but $\min\{X,M\}$ has expectation $\Theta(\eps W)$.
The correct operation is a mean-preserving contraction of each survival function, as for $\E[\max]$ in~\citet{liu2017}: for $t<M$ one adds $\tau_i/M$ to $p_i(t)$, where $\tau_i=\E[(X_i-M)_+]$.
The mass added on $[0,M)$ equals the tail mass $\tau_i$ whenever the survival stays at most $1$.
Lipschitzness of $p\mapsto\E[\min(k,N(p))]$ turns that identity into an $O(c^2\eps)\,W$ change of the objective.
A grid of spacing $\delta=\eps W$ then costs at most $k\delta=k\eps W$ pathwise and leaves $h=1/\eps^2$ layers.
Layer cake becomes a sum of $h$ occupancy terms, and a per-layer error of $O(\eps^2)$ sums to $O(\eps)\,W$ because $\delta h=W/\eps$.
The grid term is the source of the factor $k$ in the relative error $O((c^2+k)\eps)$.

\paragraph{Signatures for occupancy, not for the void probability.}
This is the point at which the algorithm parts company with the $\E[\max]$ literature.
\citet{chen2016} sketch a set by quantities that determine $\Pr[N=0]=\prod(1-q_i)$.
For $\Topk$ the layer-$r$ contribution is $f(p)=\E[\min(k,N(p))]$, the whole truncated occupancy distribution, not a single void probability.
We split probabilities at $\eta=\eps^2/(k\Lambda)$ with $\Lambda=2k+C_0\log(k/\eps)$.
Big coordinates are few on unsaturated layers (Chernoff: once the mean exceeds $\Lambda$, one may replace $f$ by $k$), so they can be histogrammed into $B=O(k\Lambda^2/\eps^4)$ buckets.
Tiny coordinates are replaced by a Poisson of the same total mass, at Le Cam cost $k\sum p_i^2\le k\eta\mu\le\eps^2$.
Additivity of the sketch forces the tiny-mass grid to have width $\gamma=\eps^2/n$: $n$ independent rounding errors of size $\gamma$ sum to $\eps^2$.
That $n$ in the denominator produces $n^{O(h)}=n^{O(1/\eps^2)}$ signatures from the tiny counts, independently of $k$; the histograms contribute the extra factor $(k/\eps)^{O(k\Lambda^2/\eps^6)}$.

\paragraph{Two ways to realize a signature.}
A signature is an integer vector of magnitude $M=\mathrm{poly}(n,1/\eps)$ (or $n^{O(1/\eps^2)}$ for occupancy counts).
On a fixed-$d$ packing family the decoder is realized approximately by the packing LP, after enumerating $n^{f(d,1/\eps)}$ heavy items (Section~\ref{sec:packing}).
On a query-weight family the same vector is realized exactly by one mixed-radix exact-sum call, time $\poly(\langle\F\rangle)\,M^{O(1)}$ (Section~\ref{sec:exactsum}).
The combinatorial structure of $\F$ never enters the stochastic analysis.
Both boundaries inherit the constant-factor transfer: packing via a deterministic max-sum PTAS, exact-sum via Proposition~\ref{prop:es-to-ms}.
Cuts have the transfer and neither boundary.

\paragraph{A second sketch when occupancy concentrates.}
The $k$ in the relative error of Theorem~\ref{thm:main} is an artifact of the uniform $\eps W$-grid, and the $k^{\mathrm{poly}(k)}$ factor is an artifact of unsaturated-layer histograms.
Both become unnecessary once $k\ge C/\eps^2$.
Pathwise $\Topk(x)=\min_{t\ge 0}\{kt+\sum_i(x_i-t)_+\}$, so the surplus proxy $A_T=kT+\sum_i(X_i-T)_+$ always dominates $\Topk$.
Every set $S$ has a mixture quantile $T_S$ at which $\sum_{i\in S}\Pr[X_i>T_S]\le k\le\sum_{i\in S}\Pr[X_i\ge T_S]$.
The expected one-sided occupancy deviation is $O(\sqrt{k})$.
Thinning the indicators below $T$ and a leave-one-out identity for the unbounded surplus above $T$ give
\[
\E[A_{T_S}(S)]
\le
\bigl(1+O(k^{-1/2})\bigr)\,\E[\Topk(S)].
\]
Thus $k=\Omega(1/\eps^2)$ already yields a $(1-O(\eps))$ proxy.
The unknown $T_{S^*}$ equals one of $O(\eps^{-1}\log(n/\eps))$ grid heights after a multiplicative rounding, and each height is tried.
At each threshold the triple $(\Pr[X_i>T],\Pr[X_i\ge T],\E[(X_i-T)_+])$ is a three-dimensional additive type.
A single truncated mean at a fixed $T$, as in the $k=1$ surplus search, does not give this comparison: the mixture quantile is chosen so that occupancy sits at $k$, and the resulting three-dimensional type is what the packing or covering LP realizes.
The occupancy scheme at accuracy $\eps\leftarrow\Theta(\eps/k)$ is already a PTAS for $k\le C/\eps^2$, so the two sketches together are a PTAS for every $k$, once a realization primitive is supplied.
The mixture-quantile proxy is the new large-$k$ object: it is not a void probability at a larger $k$, and the same three-dimensional type can sit on any additive realization primitive (a fixed-dimensional LP or exact-sum).

\paragraph{Hardness by sparse lotteries.}
The matching negative statements use a single encoding, already at $k=1$.
Independent copies $X_i=B(w_i,q)$ satisfy $q(1-q)^{n-1}w(S)\le\E[\max(S)]\le q\,w(S)$, so a $\rho$-approximation for $\E[\Topk]$ is a $\rho(1-q)^{n-1}$-approximation for max-sum on the same $\F$.
Choosing $q$ smaller than the target gap transfers independent-set inapproximability and max-cut $\mathrm{APX}$-hardness, as uniform lower bounds below packing.
The same encoding on two-dimensional knapsack---a packing family---rules out a generic FPTAS or EPTAS on the packing class itself.
The construction does not require $k>1$.

\section{Model and main theorems}
\label{sec:model}

Let $X_1,\dots,X_n$ be independent nonnegative discrete random variables, each given by an explicit list of atoms and masses.
Let $\F\subseteq 2^{[n]}$ be a down-set, an up-set, or an arbitrary membership family.
The main maximization ladder is Theorem~\ref{thm:transfer} (constant factor from max-sum), Theorem~\ref{thm:noptas} (no generic PTAS from max-sum alone), Theorem~\ref{thm:packing} (PTAS on fixed-$d$ packing), and Theorem~\ref{thm:limits} (no generic FPTAS or EPTAS on that same packing class).
The minimization ladder is Theorem~\ref{thm:min-transfer}, Theorem~\ref{thm:covering}, and Theorem~\ref{thm:min-limits}.
Query-weight exact-sum is a second sufficient condition for a PTAS in both directions, developed in Section~\ref{sec:exactsum}.
The integer $k\ge 1$ is part of the input; every bound below is written with $k$ kept explicit.

\begin{problem}[Non-adaptive stochastic Top-$k$ sum]
\label{pr:topk}
Input: the laws of $X_1,\dots,X_n$, a family $\F$, an integer $k\ge 1$, and $\eps\in(0,1/2)$.
Output: $S\in\F$ with $\E[\Topk(S)]\ge(1-O(\eps))\OPT$.
\end{problem}

\begin{problem}[Non-adaptive stochastic Top-$k$ sum, minimization]
\label{pr:min-topk}
Input: the same data as Problem~\ref{pr:topk}.
Output: $S\in\F$ with $\E[\Topk(S)]\le(1+O(\eps))\,\OPT_{\min}$, where $\OPT_{\min}=\min_{S'\in\F}\E[\Topk(S')]$.
\end{problem}

\begin{definition}[Max-sum and min-sum]
\label{def:max-sum}
Given nonnegative weights $w_i$, max-sum returns $S\in\F$ maximizing $\sum_{i\in S}w_i$ (or an $\alpha$-approximate such set, $\alpha\le 1$).
Min-sum returns a minimizer (or an $\alpha$-approximate such set, $\alpha\ge 1$).
\end{definition}

\begin{definition}[Fixed-dimensional packing]
\label{def:packing}
A family $\F$ is a $d$-dimensional packing family if
\[
\F=\Bigl\{S\subseteq[n]:\sum_{i\in S}a_{ji}\le b_j\text{ for all }j=1,\dots,d\Bigr\}
\]
with $a_{ji},b_j\ge 0$ encoded in binary and with $d$ independent of $n$.
The natural LP is $\{x\in[0,1]^n:Ax\le b\}$, where \(A=(a_{ji})\) and \(b=(b_j)\).
One- and two-dimensional knapsack are the cases $d=1$ and $d=2$.
Cardinality is the case $d=1$ with $a_{1i}=1$.
\end{definition}

\begin{definition}[Fixed-dimensional covering]
\label{def:covering}
A family $\F$ is a $d$-dimensional covering family if
\[
\F=\Bigl\{S\subseteq[n]:\sum_{i\in S}a_{ji}\ge b_j\text{ for all }j=1,\dots,d\Bigr\}
\]
with $a_{ji},b_j\ge 0$ encoded in binary and with $d$ independent of $n$.
The natural LP is $\{x\in[0,1]^n:Ax\ge b\}$.
One- and two-dimensional covering knapsack are the cases $d=1$ and $d=2$.
A lower cardinality bound $|S|\ge m$ is the case $d=1$ with $a_{1i}=1$.
The family is an up-set: if $S\in\F$ and $S\subseteq T$ then $T\in\F$.
\end{definition}

\begin{definition}[Query-weight exact-sum]
\label{def:exact-sum}
Let $\langle\F\rangle$ denote the bit length of the encoding of $\F$ (the ground set, the combinatorial description, and any numeric parameters written in binary).
Given integer \emph{query} weights $w_i$ of magnitude $M=\max_i|w_i|$ and a target $t$, decide whether some $S\in\F$ satisfies $\sum_{i\in S}w_i=t$ (and return one such $S$ if so).
A vector target is encoded by mixed radix.
The oracle is \emph{query-weight pseudopolynomial} if it runs in time $\poly(\langle\F\rangle)\,M^{O(1)}$.
The magnitude $M$ is that of the query, not of numeric parameters of $\F$: a knapsack DP in time $O(nCM)$ has this form only when the capacity $C$ is already bounded by $\poly(\langle\F\rangle)$ (unary capacities, or $C=\poly(n)$).
\end{definition}

Cardinality, DAG $s$--$t$ paths, and several tree / matching families admit query-weight exact-sum (subset-sum or path DP in the query magnitude).
Binary knapsack and two-dimensional knapsack do not, unless their capacities are encoded in unary.
General cuts do not: exact-sum is already $\mathrm{NP}$-hard for $0$-$1$ weights.

\begin{proposition}[Exact-sum implies a max-sum FPTAS]
\label{prop:es-to-ms}
Assume the query weights are nonnegative.
If $\F$ admits query-weight exact-sum, then deterministic max-sum on $\F$ admits an FPTAS: for every $\eps\in(0,1)$ one can return $S\in\F$ with
\[
w(S)\ge(1-\eps)\max_{S'\in\F}w(S')
\]
in time $\poly(\langle\F\rangle,n,1/\eps)$.
\end{proposition}

\begin{theorem}[Constant factor]
\label{thm:transfer}
Let $k\ge 1$ and $\eps\in(0,1/2)$.
\begin{enumerate}[leftmargin=1.4em,itemsep=2pt]
\item
No single polynomial algorithm achieves a constant factor on every membership family $\F$.
Independent sets of an undirected graph are a witness, already for $k=1$.
\item
If max-sum on $\F$ admits an $\alpha$-approximation ($\alpha\le 1$), Algorithm~\ref{alg:transfer} returns $S\in\F$ with
\[
\E[\Topk(S)]
\ge
\frac{\alpha}{(1+\alpha)(1+\eps)}\,\OPT
\]
after $O\bigl(\eps^{-1}\log(\mathrm{range})\bigr)$ oracle calls, where $\mathrm{range}$ is the ratio of the largest atom to the smallest positive truncated mean that can occur.
\end{enumerate}
\end{theorem}

Independence is used only for the lower bound on the output set.
Every fixed-$d$ packing family admits a PTAS for deterministic max-sum~\citep{frieze1984,ibarra1975}, hence a constant-factor max-sum oracle, so part~(ii) applies to the whole class.
If $\F$ instead admits query-weight exact-sum, Proposition~\ref{prop:es-to-ms} yields a max-sum FPTAS, so part~(ii) applies on that class as well.
The reduction in that proposition rounds weights to magnitude $O(n/\eps)$ and therefore does not place exact max-sum in $\mathrm{P}$ for binary weights of bit length $L$.

\begin{theorem}[No generic PTAS]
\label{thm:noptas}
No single PTAS works on every membership family that already has a constant-factor max-sum approximation.
Cuts are a witness: they have a $0.878$-approximation for max-sum~\citep{goemans1995}, but unless $P=NP$ there is no PTAS for $\E[\Topk]$ on every cut family, already for $k=1$.
\end{theorem}

A constant-factor max-sum oracle is therefore not a sufficient generic hypothesis for a PTAS.

\begin{theorem}[PTAS on packing]
\label{thm:packing}
Let $d\ge 1$ be a fixed integer.
If $\F$ is a $d$-dimensional packing family, Problem~\ref{pr:topk} admits a PTAS for every $k\ge 1$, running in time $n^{f(d,1/\eps)}$ for a function $f$ independent of $n$.
\end{theorem}

Section~\ref{sec:packing} constructs the scheme from the occupancy decoder of Theorem~\ref{thm:main} when $k\le C/\eps^2$ and the quantile type of Theorem~\ref{thm:quantile} when $k\ge C/\eps^2$.
The scale $W$ is obtained from Algorithm~\ref{alg:transfer} with a packing max-sum PTAS, not from an exact surplus maximizer.
In the theorems of this section $\eps$ is the requested accuracy; the fold, the grids, and the LPs use a smaller internal accuracy $\eps_{\mathrm{int}}=\Theta(\eps/k)$ (Corollary~\ref{cor:ptas}), written $\eps$ again in Sections~\ref{sec:packing} and~\ref{sec:covering}.

\begin{theorem}[PTAS on covering]
\label{thm:covering}
Let $d\ge 1$ be a fixed integer.
If $\F$ is a $d$-dimensional covering family, Problem~\ref{pr:min-topk} admits a PTAS for every $k\ge 1$, running in time $n^{f(d,1/\eps)}$ for a function $f$ independent of $n$.
\end{theorem}

Section~\ref{sec:covering} is the covering counterpart of Section~\ref{sec:packing}: the same signatures, a min-sum PTAS for the scale $W$, and a covering LP in which fractionals are raised to $1$.
Folding is not symmetric (Lemma~\ref{lem:fold-back}): a set that is cheap after the fold need not have been cheap originally.

\begin{theorem}[Occupancy scheme]
\label{thm:main}
Let $k\ge 1$ be an integer and $\eps\in(0,1/2)$.
Write $\Lambda:=2k+C_0\log(k/\eps)$ with $C_0$ as in Lemma~\ref{lem:pb-chernoff}.
If $\F$ admits query-weight exact-sum, there is an algorithm that returns $S\in\F$ with
\[
\E[\Topk(S)]
\ge \OPT-O((k+1)\eps)\,W
\ge\bigl(1-O((k+1)\eps)\bigr)\OPT,
\]
where $W$ is the surplus scale of Theorem~\ref{thm:surplus},
in time
\[
n^{O(1/\eps^2)}\,(k\Lambda/\eps)^{O(k\Lambda^2/\eps^6)}
\]
times the cost of one query-weight exact-sum call of the same magnitude.
\end{theorem}

\begin{corollary}[When the scheme is a PTAS]
\label{cor:ptas}
Let $\eps':=\eps/(C(k+1))$ for a large enough absolute constant $C$.
Theorem~\ref{thm:main} at accuracy $\eps'$ returns a $(1-O(\eps))$-approximation in time
\[
n^{O(k^2/\eps^2)}\,(k\Lambda'/\eps)^{O\bigl(k^7\Lambda'^2/\eps^6\bigr)},
\]
where $\Lambda':=2k+C_0\log(k/\eps')$.
For every fixed $\eps\in(0,1/2)$ this time is polynomial in $n$ whenever $k=O(1/\eps^2)$: the $n$-exponent is then $O(1/\eps^6)$ and the remaining factor depends only on $\eps$.
\end{corollary}

The $n$-exponent still depends on $\eps$, so the algorithm is not an EPTAS; the dependence on $1/\eps$ is not polynomial, so it is not an FPTAS.
Theorem~\ref{thm:quantile} covers the complementary range $k=\Omega(1/\eps^2)$.

\begin{theorem}[Quantile PTAS]
\label{thm:quantile}
Let $k\ge 1$ and $\eps\in(0,1/8)$.
There is an absolute constant $C\ge 2$ such that the following holds.
If $\F$ admits query-weight exact-sum, and if $k\ge C/\eps^2$,
then Algorithm~\ref{alg:quantile} returns $S\in\F$ with $\E[\Topk(S)]\ge(1-O(\eps))\OPT$ after
\[
\poly\bigl(n,1/\eps\bigr)
\]
query-weight exact-sum calls, each of magnitude $\poly(n,1/\eps)$.
\end{theorem}

\begin{corollary}[PTAS for every $k$]
\label{cor:allk}
If $\F$ admits query-weight exact-sum, Problem~\ref{pr:topk} admits a PTAS for every $k\ge 1$.
For $k\le C/\eps^2$ the running time is $n^{O(1/\eps^6)}$ times a factor depending only on $\eps$, times $\poly(\langle\F\rangle)$; for $k\ge C/\eps^2$ it is $\poly(n,1/\eps)$ query-weight exact-sum calls of magnitude $\poly(n,1/\eps)$.
\end{corollary}

\begin{corollary}[Minimization PTAS]
\label{cor:min-allk}
If $\F$ admits query-weight exact-sum, then $\min_{S\in\F}\E[\Topk(S)]$ admits a PTAS for every $k\ge 1$, with the same time bounds as Corollary~\ref{cor:allk}.
\end{corollary}

The scale $W$ is $\E[\Topk(S_0)]$ from Algorithm~\ref{alg:min-transfer} (min-sum via Proposition~\ref{prop:es-to-ms}).
The decoder is Algorithms~\ref{alg:ptas} and~\ref{alg:quantile} with this $W$, returning the realized set of smallest $\E[\Topk]$ on the folded laws.
A min-optimum $S^*$ satisfies $\E[\Topk(S^*)]\le W$, so Theorem~\ref{thm:disc} applies with $c=1$.
Lemma~\ref{lem:fold-back} converts a folded value at most $(1+4\eps)W$ into an original value at most $U/(1-4\eps)\le(1+5\eps)U$.
Downward-closed packing is excluded because the minimum is the empty set.
Fixed-$d$ covering is Theorem~\ref{thm:covering}; $s$--$t$ paths, spanning trees, and assignment sit on this exact-sum boundary and are neither packing nor covering~\citep{ibrahimpur2020}.

\begin{remark}[Uniform versus per-instance]
\label{rem:uniform}
The negatives are uniform.
There is no single constant on every membership family, no single PTAS on every constant-factor max-sum family, and no single FPTAS or EPTAS on every fixed-$d$ packing family (Theorems~\ref{thm:transfer}(i),~\ref{thm:noptas} and~\ref{thm:limits}).
For minimization there is no single constant on every membership family, no single PTAS from a min-sum constant on an arbitrary up-set (vertex cover), and no single FPTAS on every fixed-$d$ covering family (Theorems~\ref{thm:min-transfer},~\ref{thm:covering} and~\ref{thm:min-limits}).
Independent sets, cuts, two-dimensional knapsack, and two-dimensional covering knapsack are witnesses.
A classification of which concrete $\F$ necessarily admit no constant-factor algorithm, or no PTAS, is open (Section~\ref{sec:concl}).
On fixed-$d$ packing the PTAS matches the FPTAS/EPTAS lower bounds; on fixed-$d$ covering it matches an FPTAS lower bound.
Neither class is the only source of a PTAS (Section~\ref{sec:exactsum}).
\end{remark}

\section{Preliminaries}
\label{sec:prelim}

Write $B(a,p)$ for a random variable that equals $a$ with probability $p$ and $0$ otherwise, and $\Tr(x,T)=\max\{x-T,0\}$.
If $X$ has atoms $0=a_0<a_1\le\cdots$ with masses $p_j$, there exist independent $Z_j=B(a_j,p_j/\sum_{j'\le j}p_{j'})$ with $X=\max_j Z_j$ in law.

\begin{lemma}[Layer cake]
\label{lem:layer}
For $x\in\R_+^m$,
\begin{equation}
\label{eq:layer}
\Topk(x)
=\int_0^\infty\min\bigl(k,\,\#\{i:x_i>t\}\bigr)\,dt.
\end{equation}
If every $x_i$ lies in $\{0,\delta,\dots,h\delta\}$, then
$\Topk(x)=\delta\sum_{r=0}^{h-1}\min(k,\#\{i:x_i>r\delta\})$.
\end{lemma}

\begin{lemma}[Poisson-binomial truncation]
\label{lem:pb}
If $N$ is Poisson-binomial with mean $\mu$, then $\E[\min(k,N)]\ge k\mu/(k+\mu)$ (the right-hand side is $0$ when $\mu=0$).
\end{lemma}

\begin{lemma}[$\ell_1$ Lipschitz]
\label{lem:f-lip}
Let $f(p)=\E[\min(k,N(p))]$ for a Poisson-binomial $N(p)$.
Then $\partial f/\partial p_i=\Pr[N^{(-i)}<k]\in[0,1]$, so $f$ is nondecreasing and concave and $\lvert f(p)-f(q)\rvert\le\|p-q\|_1$.
Also $0\le f(p)\le\min(k,\sum_i p_i)$.
\end{lemma}

\begin{lemma}[Chernoff saturation~\citep{chernoff1952}]
\label{lem:pb-chernoff}
If $N$ is Poisson-binomial with mean $\mu$, then $\Pr[N\le\mu-t]\le\exp(-t^2/(2\mu))$ for $0\le t\le\mu$, and $\Pr[N\ge\mu+t]\le\exp(-t^2/(2(\mu+t/3)))$ for $t\ge 0$.
There is an absolute constant $C_0$ such that, writing $\Lambda:=2k+C_0\log(k/\eps)$, the inequality $\mu\ge\Lambda$ implies $k-f(p)\le\eps^2$.
\end{lemma}

\begin{lemma}[Le Cam~\citep{lecam1960}]
\label{lem:lecam}
If $N$ is Poisson-binomial with parameters $p_i$ and mean $\mu$, and $Z\sim\mathrm{Poisson}(\mu)$, then $d_{\mathrm{TV}}(N,Z)\le\sum_i p_i^2$.
If $M$ is independent of both and $g(x)=\min(k,x)$, then $\lvert\E[g(N+M)]-\E[g(Z+M)]\rvert\le k\sum_i p_i^2$.
\end{lemma}

\begin{lemma}[Poisson mean]
\label{lem:pois-mean}
If $Z_\lambda\sim\mathrm{Poisson}(\lambda)$ and $M$ is independent of $Z_\lambda$, then $\lambda\mapsto\E[\min(k,Z_\lambda+M)]$ is $1$-Lipschitz.
\end{lemma}

Proofs of Lemmas~\ref{lem:layer}--\ref{lem:pois-mean} are in Appendix~\ref{app:prelim}.

\section{Surplus Constant Approximation}
\label{sec:surplus}
\label{sec:transfer}

\begin{definition}[Surplus fixed point]
\label{def:fp}
A pair $(S^*,T)$ with $S^*\in\F$ and $T\ge 0$ is a surplus fixed point if
\begin{equation}
\label{eq:FP}
\max_{S\in\F}\sum_{i\in S}\E[\Tr(X_i,T)]
=\sum_{i\in S^*}\E[\Tr(X_i,T)]
=kT.
\end{equation}
\end{definition}

The map $g(T)=\max_{S\in\F}\sum_{i\in S}\E[\Tr(X_i,T)]$ is continuous and nonincreasing, $g(0)\ge 0$, and $g(T)=0$ for all large $T$, so $g(T)-kT$ has a root.

\begin{lemma}[Deterministic sandwich]
\label{lem:det-sand}
For $x\in\R_+^m$ and $T\ge 0$, $\Topk(x)\le kT+\sum_i\Tr(x_i,T)$.
If additionally $\sum_i\Tr(x_i,T)\ge kT$, then $\Topk(x)\ge kT$.
\end{lemma}

\begin{lemma}[Decreasing tail]
\label{lem:tail}
Let $\mu:[0,\infty)\to[0,\infty)$ be nonincreasing and let $\int_T^\infty\mu=kT$.
Then $\int_0^\infty k\mu/(k+\mu)\,dt\ge kT$.
\end{lemma}

\begin{lemma}[One set]
\label{lem:one-set}
If $X_1,\dots,X_m$ are independent and nonnegative and $\sum_i\E[\Tr(X_i,T)]=kT$, then $\E[\Topk(X)]\ge kT$.
The same lower bound holds if the surplus is at least $kT$: there exists $T'\ge T$ at which equality holds for this set, hence $\E[\Topk]\ge kT'\ge kT$.
\end{lemma}

\begin{theorem}
\label{thm:surplus}
At a surplus fixed point $(S^*,T)$,
\[
kT\le\E[\Topk(S^*)]\le\OPT\le 2kT.
\]
In particular $S^*$ is a $2$-approximation.
\end{theorem}

For $k=1$ the statement is the truncation sandwich of~\citet{liu2017}.

For two-point laws, maximizing surplus at a fixed $T$ is max-sum with weights $p_i\Tr(a_i,T)$.
A binary search on $T$ finds a fixed point to any $\eps$-accuracy; the PTAS only needs some $W$ with $W\le\OPT\le 2W$ (take $W=kT$).

The $2$-approximation of Theorem~\ref{thm:surplus} assumes an exact surplus maximizer.
If max-sum is available only to within $\alpha\le 1$, the decreasing search of~\citet{kleinberg2000,guha2009} used for $\E[\max]$ in~\citet[Thm.~3.9]{liu2017} applies verbatim, with threshold $kT$ in place of $T$.
The ratio does not depend on $k$.
If the family instead has query-weight exact-sum, Proposition~\ref{prop:es-to-ms} yields a max-sum FPTAS and the same search applies.

\begin{algorithm}[t]
\caption{Max-sum transfer for Top-$k$ sum}
\label{alg:transfer}
\begin{algorithmic}[1]
\Require discrete laws, family $\F$, integer $k\ge 1$, $\eps\in(0,1/2)$, $\alpha$-approximate max-sum oracle $\mathcal{A}$
\State $T\leftarrow \max_i\max\supp(X_i)$
\State $S\leftarrow \mathcal{A}\bigl(\E[\Tr(X_1,T)],\dots,\E[\Tr(X_n,T)],\F\bigr)$
\While{$\E[\Tr(S,T)]\le kT$ and $T>0$}
  \State $T\leftarrow T/(1+\eps)$
  \State $S\leftarrow \mathcal{A}\bigl(\E[\Tr(X_1,T)],\dots,\E[\Tr(X_n,T)],\F\bigr)$
\EndWhile
\State \Return $S$
\end{algorithmic}
\end{algorithm}

When $\alpha=1$ the ratio is $1/(2(1+\eps))$, which matches Theorem~\ref{thm:surplus} up to the $(1+\eps)$ slack of the geometric search.

\begin{corollary}[Cuts]
\label{cor:cuts}
If $\F$ is the family of cuts $\{\delta(U)\}$ and the edge laws are independent,
Goemans--Williamson~\citep{goemans1995} ($\alpha\approx 0.878$) yields ratio
\[
\frac{\alpha}{(1+\alpha)(1+\eps)}\approx \frac{0.878}{1.878\,(1+\eps)}.
\]
\end{corollary}

The same family has no PTAS (Theorem~\ref{thm:noptas}).

Further instantiations are collected in Section~\ref{sec:inst}.

The minimization problem $\min_{S\in\F}\E[\Topk(S)]$ is the same comparison with the search reversed, as for $\E[\max]$ in~\citet{liu2017}.
On a downward-closed family the minimum is the empty set; the interesting instances are covering-type.

\begin{algorithm}[t]
\caption{Min-sum transfer for Top-$k$ sum}
\label{alg:min-transfer}
\begin{algorithmic}[1]
\Require discrete laws, family $\F$, integer $k\ge 1$, $\eps\in(0,1/2)$, $\alpha$-approximate min-sum oracle $\mathcal{A}$ ($\alpha\ge 1$)
\State $T\leftarrow k^{-1}\min\{\E[X_i]:\E[X_i]>0\}$ (or halt if every law is identically zero)
\State $S\leftarrow \mathcal{A}\bigl(\E[\Tr(X_1,T)],\dots,\E[\Tr(X_n,T)],\F\bigr)$
\While{$\E[\Tr(S,T)]\ge \alpha kT$ and $T\le\max_i\max\supp(X_i)$}
  \State $T\leftarrow T(1+\eps)$
  \State $S\leftarrow \mathcal{A}\bigl(\E[\Tr(X_1,T)],\dots,\E[\Tr(X_n,T)],\F\bigr)$
\EndWhile
\State \Return $S$
\end{algorithmic}
\end{algorithm}

\begin{theorem}[Constant factor for minimization]
\label{thm:min-transfer}
Let $k\ge 1$ and $\eps\in(0,1/2)$.
If min-sum on $\F$ admits an $\alpha$-approximation ($\alpha\ge 1$), Algorithm~\ref{alg:min-transfer} returns $S\in\F$ with
\[
\E[\Topk(S)]
\le
(1+\alpha)(1+\eps)\,\OPT_{\min},
\]
where $\OPT_{\min}=\min_{S'\in\F}\E[\Topk(S')]$, after $O\bigl(\eps^{-1}\log(\mathrm{range})\bigr)$ oracle calls.
In particular $W:=\E[\Topk(S)]$ satisfies $\OPT_{\min}\le W\le(1+\alpha)(1+\eps)\,\OPT_{\min}$.
\end{theorem}

Every fixed-$d$ covering family admits a PTAS for deterministic min-sum (the covering analogue of~\citet{frieze1984}; one-dimensional covering knapsack has an FPTAS), hence a constant-factor min-sum oracle, so the theorem applies to the whole class.
The same rounding as Proposition~\ref{prop:es-to-ms}, scanning the smallest feasible target, is a min-sum FPTAS.
Thus every query-weight family already has an exact enough scale $W$ for the decoder of Corollary~\ref{cor:min-allk}.
Proofs are in Appendix~\ref{app:surplus}.

\section{Reducing the support}
\label{sec:disc}

Let $W>0$ satisfy $\E[\Topk(X)]\le cW$ with $c\eps<\min\{1/4,\,k/2\}$ (Theorem~\ref{thm:surplus} gives $c=2$ and $W=kT$).
Set $M:=W/\eps$.

\begin{remark}[Naive cap]
\label{rem:naive}
If $Y_i=\min\{X_i,M\}$ and $X_1=B(H,p)$ with $pH=\Theta(W)$ and $H\gg M$, then $\E[\Topk(X)]=\Theta(W)$ while $\E[\Topk(Y)]=\Theta(\eps W)$.
The tail must be folded, not deleted.
\end{remark}

\begin{definition}[Mean-preserving fold]
\label{def:fold}
Write $p_i(t)=\Pr[X_i>t]$ and $\tau_i:=\int_M^\infty p_i(t)\,dt=\E[(X_i-M)_+]$.
The folded survival is
\[
\bar p_i(t)
:=
\begin{cases}
0 & t\ge M,\\
\min\bigl(1,\,p_i(t)+\tau_i/M\bigr) & t<M.
\end{cases}
\]
Let $\bar X_i$ be any nonnegative random variable with this survival function.
Then $\supp(\bar X_i)\subseteq[0,M]$.
\end{definition}

Write $N(t)$, $\bar N(t)$, and $\mu(t)$ for the occupancy counts and mean of $(X_i)_i$ and $(\bar X_i)_i$.
For $t\ge M$ one has $\bar N(t)=0$.
For $t<M$, $\bar N(t)$ stochastically dominates $N(t)$.
If the minimum in the definition does not bind, then $\int_0^M(\bar p_i-p_i)\,dt=\tau_i$.

\begin{lemma}[Tail mass]
\label{lem:tailmass}
Let $\Sigma:=\int_M^\infty\mu(t)\,dt$.
Then $\Sigma\le 2cW$.
\end{lemma}

\begin{lemma}[Fold]
\label{lem:fold}
$\bigl\lvert\E[\Topk(X)]-\E[\Topk(\bar X)]\bigr\rvert\le O(c^2\eps)\,W$.
\end{lemma}

\begin{lemma}[Grid]
\label{lem:grid}
If $\supp(X_i)\subseteq[0,M]$, the rounding $\tilde X_i=\delta\lfloor X_i/\delta\rfloor$ with $\delta=\eps W$ satisfies $0\le\Topk(X)-\Topk(\tilde X)<k\eps W$ pathwise, and each $\tilde X_i$ takes values in a set of size $h=1/\eps^2$.
\end{lemma}

\begin{theorem}
\label{thm:disc}
After the fold of Definition~\ref{def:fold} and the grid of Lemma~\ref{lem:grid},
\[
\bigl\lvert\E[\Topk(X)]-\E[\Topk(\tilde X)]\bigr\rvert
\le O\bigl((c^2+k)\eps\bigr)W.
\]
Every $\tilde X_i$ is supported on $\{0,\delta,\dots,M\}$ with $h=1/\eps^2$ atoms.
\end{theorem}

Relative to $\OPT\ge W$ the error is $O((c^2+k)\eps)$; it becomes $O(\eps)$ after the substitution $\eps\leftarrow\Theta(\eps/(c^2+k))$ of Corollary~\ref{cor:ptas}.

\begin{lemma}[Fold-back]
\label{lem:fold-back}
Assume Definition~\ref{def:fold} with $M=W/\eps$ and $\eps\in(0,1/20)$.
If $\E[\Topk(\bar X)]\le U$ and $U\le(1+4\eps)W$, then
\[
\E[\Topk(X)]
\le
\frac{U}{1-4\eps}
\le
(1+5\eps)\,U.
\]
\end{lemma}

Proofs are in Appendix~\ref{app:disc}.

\section{Type-histogram signatures}
\label{sec:sig}

Assume every $X_i$ is already folded and gridded, and write $\delta=\eps W$, $h=1/\eps^2$.
By Lemma~\ref{lem:layer},
\[
\E[\Topk(S)]
=\delta\sum_{r=0}^{h-1}f\bigl(p_{\cdot,r}(S)\bigr),
\qquad
p_{i,r}:=\Pr[X_i>r\delta],
\qquad
f(p):=\E[\min(k,N(p))].
\]
Because $\delta h=W/\eps$, a per-layer error of $O(\eps^2)$ sums to $O(\eps)\,W$.

With $\Lambda$ from Lemma~\ref{lem:pb-chernoff}, set
\begin{align*}
\eta&:=\frac{\eps^2}{k\Lambda},&
\gamma&:=\frac{\eps^2}{n},&
\beta&:=\frac{\eps^2\eta}{2\Lambda}=\frac{\eps^4}{2k\Lambda^2}.
\end{align*}
Call $p$ \emph{tiny} if $0<p\le\eta$ and \emph{big} if $p>\eta$.
Partition $(\eta,1]$ into $B\le 1/\beta$ intervals $I_b$ of length at most $\beta$, and write $\hat p_b$ for the right endpoint of $I_b$.
If $\mu<\Lambda$ then the number of big coordinates is at most $\Lambda/\eta=k\Lambda^2/\eps^2$, so rounding them costs at most $(\Lambda/\eta)\beta=\eps^2/2$ in $\ell_1$.

\begin{definition}[Signature]
\label{def:sg}
For each $i$ and layer $r$,
\[
\tau_{i,r}:=\mathbf{1}_{p_{i,r}\le\eta}\Bigl\lfloor\frac{p_{i,r}}{\gamma}\Bigr\rfloor,
\qquad
H_{i,r,b}:=\mathbf{1}_{p_{i,r}\in I_b}.
\]
Set $\Sg(X_i)=\bigl((\tau_{i,r})_r,(H_{i,r,b})_{r,b}\bigr)$ and $\Sg(S)=\sum_{i\in S}\Sg(X_i)$.
Write $\sigma_r(S)=\sum_{i\in S}\tau_{i,r}$ and $n_{r,b}(S)=\sum_{i\in S}H_{i,r,b}$.
The proxy at layer $r$ is the independent sum
\[
\widehat N_r
=\sum_b\mathrm{Bin}(n_{r,b},\hat p_b)+\mathrm{Poisson}(\gamma\sigma_r).
\]
If the proxy mean $\gamma\sigma_r+\sum_b n_{r,b}\hat p_b$ is at least $\Lambda$, replace $\E[\min(k,\widehat N_r)]$ by $k$ (saturated layer).
Set $\mathrm{Val}(\Sg(S))=\delta\sum_r\E[\min(k,\widehat N_r)]$ with that convention.
\end{definition}

The map $S\mapsto\Sg(S)$ is additive.
Encoded in mixed radix, $\Sg(S)$ is an integer of magnitude $n^{O(h)}(k\Lambda/\eps)^{O(hB)}$ with $h=1/\eps^2$ and $B\le 2k\Lambda^2/\eps^4$, hence $hB=O(k\Lambda^2/\eps^6)$.

\begin{lemma}[Unsaturated layer]
\label{lem:layer-unsat}
Fix a layer with true mean $\mu<\Lambda$ and let $(\sigma,(n_b))$ be its signature.
Then $\lvert f(p)-\E[\min(k,\widehat N)]\rvert\le 4\eps^2$.
\end{lemma}

\begin{lemma}[Saturated layer]
\label{lem:layer-sat}
If a layer is marked saturated, then $f(p)\ge k-O(\eps^2)$.
Two saturated layers differ by $O(\eps^2)$.
\end{lemma}

\begin{theorem}
\label{thm:sig}
If $\Sg(S_1)=\Sg(S_2)$, then
\[
\bigl\lvert\E[\Topk(S_1)]-\E[\Topk(S_2)]\bigr\rvert\le O(\eps)\,W,
\]
and both sides are $O(\eps)W$-close to $\mathrm{Val}(\Sg(S_1))$.
\end{theorem}

Tiny mass is quantized at width $\gamma=\eps^2/n$ so that additivity does not accumulate more than $\eps^2$ total error.
The mixed-radix encoding of tiny counts contributes $n^{O(h)}=n^{O(1/\eps^2)}$ signatures, which is why the scheme is not an EPTAS even after $k$ is held fixed.
Proofs are in Appendix~\ref{app:sig}.

\section{A quantile scheme for large \texorpdfstring{$k$}{k}}
\label{sec:quantile}

The occupancy scheme of Theorem~\ref{thm:main} has $n$-exponent $O(k^2/\eps^2)$ after substituting $\eps\leftarrow\Theta(\eps/k)$ into a uniform $\eps W$-grid, and a $k^{\mathrm{poly}(k,1/\eps)}$ factor from unsaturated-layer histograms.
Both costs disappear for $k\ge C/\eps^2$: the expected occupancy deviation at a mixture quantile is $O(\sqrt{k})$, already a relative $\eps$.

Write
\begin{equation}
\label{eq:AT}
A_T(x)
:=
kT+\sum_i\Tr(x_i,T),
\qquad
N(x,T)
:=
\#\{i:x_i>T\},
\qquad
N'(x,T)
:=
\#\{i:x_i\ge T\}.
\end{equation}
Lemma~\ref{lem:det-sand} is the inequality $A_T(x)\ge\Topk(x)$.
The next lemma records that the inequality is tight at a single $t$, and that the slack is controlled by occupancy.

\begin{lemma}[Pathwise proxy]
\label{lem:path-proxy}
For $x\in\R_+^m$ and $T\ge 0$,
\[
\Topk(x)
=\min_{t\ge 0}A_t(x).
\]
Writing $H=N(x,T)$, $M=N'(x,T)$, and $z_i=\Tr(x_i,T)$,
\begin{equation}
\label{eq:path-strong}
0
\le
A_T(x)-\Topk(x)
\le
T(k-M)_+
+\frac{(H-k)_+}{H}\sum_i z_i,
\end{equation}
where the second term is $0$ when $H=0$.
If $H\le k\le M$, then $A_T(x)=\Topk(x)$.
\end{lemma}

\begin{definition}[Mixture quantile]
\label{def:mix-q}
For a set $S$ and $T\ge 0$ write
\[
\mu_S^>(T)=\sum_{i\in S}\Pr[X_i>T],
\qquad
\mu_S^\ge(T)=\sum_{i\in S}\Pr[X_i\ge T],
\qquad
\Phi_S(T)=\sum_{i\in S}\E[\Tr(X_i,T)].
\]
Call $T$ a mixture quantile of $S$ if $\mu_S^>(T)\le k\le\mu_S^\ge(T)$.
\end{definition}

Every $S$ with $\lvert S\rvert\ge k$ has a mixture quantile: $\mu_S^>(T)$ is nonincreasing from $\sum_{i\in S}\Pr[X_i>0]$ to $0$, and $\mu_S^\ge(T)$ is the left limit at each atom.
If $\lvert S\rvert\le k$ then $\Topk(S)=\sum_{i\in S}X_i$ pathwise, which is $A_0(S)$.

\begin{lemma}[Quantile proxy]
\label{lem:c1}
Let $X_1,\dots,X_m$ be independent and nonnegative, and let $T\ge 0$ satisfy $\mu^>(T)\le k\le\mu^\ge(T)$.
Write $V=\Topk(X)$, $A_T=A_T(X)$, and $\eta_k:=k^{-1/2}+k^{-1}$.
If $k\ge 2$ then
\begin{equation}
\label{eq:c1}
0
\le
\E[A_T]-\E[V]
\le
\eta_k\,\E[A_T],
\end{equation}
hence $\E[A_T]\le(1-\eta_k)^{-1}\E[V]=(1+O(k^{-1/2}))\E[V]$.
The same bound holds with $\eta_k$ replaced by $\eta_k+O((\delta+1)/k)$ after enlarging the window to $\mu^>\le k+\delta$ and $\mu^\ge\ge k-\delta$.
\end{lemma}

The error is $O(\eps)\,\E[V]$ once $k\ge C/\eps^2$.
A mixture quantile of $S^*$ equals some input atom, or after Lemma~\ref{lem:mult-round} some grid height.
At each candidate $T$ the decoder matches residual types in the slice $\mu^>\le k\le\mu^\ge$.
Proofs of Lemmas~\ref{lem:path-proxy}--\ref{lem:mult-round} are in Appendix~\ref{app:quantile}.

The uniform grid of Lemma~\ref{lem:grid} would re-introduce a factor $k$ in the relative error.
A multiplicative rounding avoids that factor.

\begin{lemma}[Multiplicative rounding]
\label{lem:mult-round}
Let $c\ge 1$ be a fixed constant, let $W\le\E[\Topk(X)]\le cW$, and assume $\supp(X_i)\subseteq[0,W/\eps]$ for every $i$.
Set $\tau:=\eps W/n$ and replace each atom $a$ by $0$ if $a<\tau$ and by $\tau(1+\eps)^{\lfloor\log_{1+\eps}(a/\tau)\rfloor}$ otherwise.
Writing $Y$ for the rounded vector, $(1-O(\eps))\Topk(X)-O(\eps)\,W\le\Topk(Y)\le\Topk(X)$ pathwise, and every $Y_i$ is supported on a grid $G$ of size $O(\eps^{-1}\log(n/\eps))$.
(The pathwise argument does not use the numerical value of $c$; the packing bound $W\le\OPT\le 5W$ is an instance.)
\end{lemma}

Section~\ref{sec:packing} realizes the triple approximately in the packing LP; Section~\ref{sec:exactsum} realizes it by exact-sum.

\section{Packing PTAS via the natural LP}
\label{sec:packing}

Theorem~\ref{thm:packing} uses the decoders of Sections~\ref{sec:sig} and~\ref{sec:quantile}, realized approximately in the packing LP.
The scale $W$ is $kT_0$ from Algorithm~\ref{alg:transfer}, run with a $(1/2)$-approximate max-sum oracle given by a packing PTAS~\citep{frieze1984,ibarra1975}.
Theorem~\ref{thm:transfer}(ii) gives $W\le\OPT\le 5W$, so Theorem~\ref{thm:disc} and Lemma~\ref{lem:mult-round} apply with an absolute constant $c=O(1)$.
Capacities are never rounded.
Throughout this section, $\eps$ is the \emph{internal} accuracy of the fold, the grids, and the LP; write $\eps_{\mathrm{out}}$ for the requested PTAS accuracy.
The substitution $\eps\leftarrow\Theta(\eps_{\mathrm{out}}/k)$ of Corollary~\ref{cor:ptas} is applied only after the construction (proof of Theorem~\ref{thm:packing}).
Quantities such as $h=1/\eps^2$ and $\eta_{\mathrm{LP}}$ are therefore expressed in the internal parameter.
The two ranges use different additive coordinates, because a basic feasible solution can leave $O(D)$ fractional items, and dropping those items is affordable only when every residual coordinate is $o(1)$ relative to the decoder's tolerance.

\paragraph{Large $k$.}
If $k\ge C/\eps^2$, the type is the triple $(\Pr[X_i>T],\Pr[X_i\ge T],\E[(X_i-T)_+])$, so $D=3$.
Guess $T$ on the multiplicative grid of Lemma~\ref{lem:mult-round} and guess a residual target $z$ for the triple of $S^*$.
An item is surplus-heavy if $\E[(X_i-T)_+]>\eta W$ with $\eta=\Theta(\eps)$.
Lemma~\ref{lem:c1} and $\E[\Topk(S^*)]\le 5W$ give $\Phi_{S^*}(T)=O(W)$, so an optimum has $O(1/\eta)$ surplus-heavies.
Occupancy coordinates of any single item are at most $1$, while those of $S^*$ are $\Theta(k)$; after surplus-heavies are fixed one therefore does not enumerate occupancy-heavies.
Solve the residual LP with $d$ packing inequalities and six signature slabs.
A BFS has at most $d+6$ fractionals.
Dropping them changes occupancy by $O(1)$ (an additive window $\delta=O(1)$ in Lemma~\ref{lem:c1}, hence a relative $O(\eps)$ once $k=\Omega(1/\eps^2)$) and surplus by $O(\eps)\,W$.
The comparison is Lemma~\ref{lem:c1} with occupancy slack $O(1)$ and surplus slack $O(\eps)\,W$.

\paragraph{Small $k$.}
If $k\le C/\eps^2$, a naive LP on all $D=O(k\Lambda^2/\eps^6)$ occupancy-histogram coordinates fails: a BFS can have $\Theta(D)$ fractionals, and dropping that many items of size $\Theta(1)$ destroys $\E[\min(k,N)]$.
The same obstruction rules out a one-hot encoding of types as LP coordinates.
The fix is to make every residual coordinate tiny and to compare residuals by Poissonization, not by $\ell_1$ on item-aligned parameter vectors.
Write $F_{\max}:=d+2h$ and $\Lambda_+:=\Lambda+F_{\max}$.
Set
\[
\eta_{\mathrm{LP}}
:=
\frac{\eps^2}{2(k\Lambda_++F_{\max})}.
\]
Guess the set $R_+$ of layers at which $S^*$ has mean at least $\Lambda_+$ ($2^{h}=2^{O(1/\eps^2)}$ guesses).
Call $i$ \emph{LP-heavy} if $\max_{r\notin R_+}p_{i,r}>\eta_{\mathrm{LP}}$.
On a layer with mean less than $\Lambda_+$ the optimum contains at most $\Lambda_+/\eta_{\mathrm{LP}}$ LP-heavies, hence at most $h\Lambda_+/\eta_{\mathrm{LP}}$ in total.
Enumerate those subsets in time $n^{O(h\Lambda_+(k\Lambda_++F_{\max})/\eps^2)}$, which is $n^{f(d,1/\eps)}$ once $k=O(1/\eps^2)$ in the internal parameter (and $n^{f(d,1/\eps_{\mathrm{out}})}$ after the substitution).
After they are fixed, every residual item satisfies $p_{i,r}\le\eta_{\mathrm{LP}}$ on every non-raised layer.
The residual LP has $d$ packing inequalities, two slabs per non-raised layer, and one covering inequality per raised layer (at most $F_{\max}$ non-box constraints; Lemma~\ref{lem:lp-drop}).
Dropping at most $F_{\max}$ fractionals changes each non-raised residual mean by at most $F_{\max}\eta_{\mathrm{LP}}\le\eps^2/2$.
Le Cam on a residual of mean $O(\Lambda_+)$ costs $k\eta_{\mathrm{LP}}\Lambda_+\le\eps^2/2$ per layer, so two residuals with equal mean have $f$-gap $O(\eps^2)$ even when they contain different items (Lemma~\ref{lem:occ-robust}).
Covering at height $\Lambda_+$ keeps every raised layer at total mean at least $\Lambda$ after the drop.

In both ranges the geometry is that of a multidimensional-knapsack PTAS~\citep{frieze1984} applied to an additive sketch.
The $n$-exponent may depend on $d$ and $1/\eps$; that is a PTAS, and Theorem~\ref{thm:limits} rules out a generic EPTAS or FPTAS on this class.
Details are in Appendix~\ref{app:packing}.

\section{Covering PTAS via the natural LP}
\label{sec:covering}

The decoder is that of Section~\ref{sec:packing}, with $Ax\ge b$ in place of $Ax\le b$.
A BFS still has at most $d+2D+C$ fractionals; raising them to $1$ preserves covering (Lemma~\ref{lem:lp-lift}) and changes each residual signature coordinate by the same budget as a drop.
The scale $W$ comes from Algorithm~\ref{alg:min-transfer} and a covering min-sum PTAS~\citep{frieze1984}, so $\OPT_{\min}\le W\le O(\OPT_{\min})$ and Theorem~\ref{thm:disc} applies to a min-optimum with $c=1$.

The only additional stochastic step is Lemma~\ref{lem:fold-back}.
The decoder returns a set whose folded value satisfies $U\le(1+4\eps)W$; the original objective is then at most $U/(1-4\eps)$.
A min-sum constant on an arbitrary up-set is not enough: vertex cover has a $2$-approximation for min-sum, but $d=\Theta(m)$.
Details are in Appendix~\ref{app:covering}.

\section{Exact-Sum Realization}
\label{sec:exactsum}
\label{sec:ptas}

Query-weight exact-sum (Definition~\ref{def:exact-sum}) is incomparable with packing: DAG $s$--$t$ paths and matchings have a query-weight DP; binary knapsack is packing and its DP is polynomial in the capacity.
This section uses the decoders of Sections~\ref{sec:sig} and~\ref{sec:quantile} and replaces approximate realization in the packing LP by exact-sum.

\begin{algorithm}[t]
\caption{Approximation scheme for Top-$k$ sum}
\label{alg:ptas}
\begin{algorithmic}[1]
\State Compute a surplus scale $W$ with $W\le\OPT\le 2W$ (Theorem~\ref{thm:surplus})
\State Fold every marginal to $[0,W/\eps]$ (Definition~\ref{def:fold}) and round to the $\eps W$-grid (Lemma~\ref{lem:grid})
\State Compute $\Sg(X_i)$ for each $i$ (Definition~\ref{def:sg})
\For{each signature $\mathrm{sg}$ of magnitude as in Section~\ref{sec:sig}}
  \State Realize some $S\in\F$ with $\Sg(S)=\mathrm{sg}$ by exact-sum, if any
\EndFor
\State \Return the realized set of largest $\E[\Topk(S)]$ (equivalently, largest $\mathrm{Val}(\mathrm{sg})$)
\end{algorithmic}
\end{algorithm}

Exact-sum returns a set with the same occupancy signature as a folded optimum, so Theorems~\ref{thm:disc} and~\ref{thm:sig} give Theorem~\ref{thm:main}.
Evaluating the true $\E[\Topk(S)]$ after realization, or the proxy $\mathrm{Val}$, are interchangeable by Theorem~\ref{thm:sig}.
The mixed-radix count of signatures is recorded in Appendix~\ref{app:exactsum}; it is the source of the $n^{O(1/\eps^2)}$ factor, and of Corollary~\ref{cor:ptas} after $\eps\leftarrow\Theta(\eps/k)$.

\begin{algorithm}[t]
\caption{Quantile scheme for large $k$}
\label{alg:quantile}
\begin{algorithmic}[1]
\Require discrete laws, family $\F$, integer $k\ge C/\eps^2$, $\eps\in(0,1/2)$
\State Compute a surplus scale $W$ with $W\le\OPT\le 2W$ (Theorem~\ref{thm:surplus})
\State Fold every marginal to $[0,W/\eps]$ (Definition~\ref{def:fold})
\State Apply the multiplicative rounding of Lemma~\ref{lem:mult-round}; let $G$ be the resulting grid
\For{each threshold $T\in G\cup\{0\}$}
  \State For each $i$, set
  $u_i\leftarrow\bigl\lfloor n\Pr[X_i>T]/\eps\bigr\rfloor$,
  $v_i\leftarrow\bigl\lfloor n\Pr[X_i\ge T]/\eps\bigr\rfloor$,
  $w_i\leftarrow\bigl\lfloor n\,\E[\Tr(X_i,T)]/(\eps W)\bigr\rfloor$
  \For{each mixed-radix target $t$ of $(u,v,w)$}
    \State Realize some $S\in\F$ with $\sum_{i\in S}(u_i,v_i,w_i)=t$ by exact-sum, if any
  \EndFor
\EndFor
\State \Return the realized set of largest $\E[\Topk(S)]$
\end{algorithmic}
\end{algorithm}

The integer vector $(u_i,v_i,w_i)$ has magnitude $\poly(n,1/\eps)$: each of $u_i,v_i$ lies in $\{0,\dots,\lfloor n/\eps\rfloor\}$ and $w_i$ lies in $\{0,\dots,\lfloor n/\eps^2\rfloor\}$ because $\E[\Tr(X_i,T)]\le W/\eps$ after the fold.
Encoded in mixed radix it is a single integer of magnitude $n^{O(1)}\eps^{-O(1)}$.
The number of targets per threshold, and the number of thresholds, are both $\poly(n,1/\eps)$.
Exact-sum matches the quantized triple of a folded optimum, so Lemma~\ref{lem:c1} yields Theorem~\ref{thm:quantile}.
Combined with Corollary~\ref{cor:ptas} on $k\le C/\eps^2$, this is Corollary~\ref{cor:allk}.
Proofs are in Appendix~\ref{app:exactsum}.

\section{No generic {EPTAS} or {FPTAS} on packing}
\label{sec:limits}

Theorem~\ref{thm:packing} is a PTAS for every $k\ge 1$ on every fixed-$d$ packing family.
It is not an EPTAS or an FPTAS.
Two-dimensional knapsack is a packing family, so the following theorem is a uniform statement about the same class: there is no FPTAS (resp.\ EPTAS) that works for every fixed-$d$ packing $\F$.

\begin{theorem}
\label{thm:limits}
Let $\F$ be the feasible sets of a two-dimensional $0$-$1$ knapsack instance (a packing family with $d=2$).
\begin{enumerate}[leftmargin=1.4em,itemsep=2pt]
\item
Unless $P=NP$, $\max_{S\in\F}\E[\Topk(S)]$ admits no FPTAS, already for $k=1$.
\item
Unless $W[1]=FPT$, the same problem admits no EPTAS, already for $k=1$.
\end{enumerate}
\end{theorem}

The proof (Appendix~\ref{app:hardness}) is the lottery encoding of Lemma~\ref{lem:lottery-transfer} with firing probability $q=\Theta(\eps/n)$, composed with the FPTAS lower bound of~\citet{magazine1984} and the EPTAS lower bound of~\citet{kulik2010}.
Cardinality for $k=1$ already admits an EPTAS~\citep{psomas2020,segev2021}.
On packing, a generic PTAS is therefore essentially optimal.

\begin{theorem}
\label{thm:min-limits}
Let $\F$ be the feasible sets of a two-dimensional $0$-$1$ covering knapsack instance (a covering family with $d=2$).
Unless $P=NP$, $\min_{S\in\F}\E[\Topk(S)]$ admits no FPTAS, already for $k=1$.
\end{theorem}

The lottery of Lemma~\ref{lem:lottery-transfer-min} transfers a deterministic FPTAS lower bound.
Lemma~\ref{lem:cover-nofptas} is that deterministic bound: an explicit Subset Sum reduction to unit-cost two-dimensional covering knapsack, with no complementation of a packing instance.
We do not invoke~\citet{kulik2010} for covering; that paper is for packing maximization.
One-dimensional covering knapsack already has a min-sum FPTAS, so it cannot be the witness.
On covering, a generic PTAS is therefore essentially optimal as far as an FPTAS is concerned.

\section{Instantiations}
\label{sec:inst}

Table~\ref{tab:inst} records the main packing ladder and the exact-sum boundary.
The first column is Theorem~\ref{thm:transfer}.
The second is Theorem~\ref{thm:packing} when $\F$ is fixed-$d$ packing.
The third is Corollary~\ref{cor:allk}.
Packing always fills the first column (max-sum PTAS). Exact-sum fills the first column by Proposition~\ref{prop:es-to-ms}.
The two PTAS columns are incomparable: binary knapsack fills only the second; DAG paths fill only the third.
Cuts fill only the first.
The minimization ladder is the same table with Theorem~\ref{thm:min-transfer}, Theorem~\ref{thm:covering}, and Corollary~\ref{cor:min-allk}: covering knapsack fills the covering-PTAS column; DAG paths fill only exact-sum; a lower cardinality bound $|S|\ge m$ is covering with $d=1$.

\begin{table}[t]
\centering
\caption{Non-adaptive $\E[\Topk]$ under concrete $\F$.
Ratios are factors $\le 1$.
``Pack.\ PTAS'' is Theorem~\ref{thm:packing}; ``ES PTAS'' is Corollary~\ref{cor:allk}.
Submodular greedy gives $1-1/e$ on cardinality, knapsack, and matroids~\citep{nemhauser1978,sviridenko2004,calinescu2011}, and does not apply to cuts.}
\label{tab:inst}
\scriptsize
\begin{tabular}{@{}p{2.6cm}p{2.8cm}p{2.2cm}p{2.6cm}@{}}
\toprule
Constraint $\F$ & Thm.~\ref{thm:transfer} & Pack.\ PTAS & ES PTAS \\
\midrule
Cardinality
  & exact, $\tfrac{1}{2(1+\eps)}$
  & \textbf{PTAS} ($d=1$)
  & \textbf{PTAS} \\
Graphic matroid
  & exact, $\tfrac{1}{2(1+\eps)}$
  & ---
  & \textbf{PTAS} if exact-weight \\
Matching
  & exact, $\tfrac{1}{2(1+\eps)}$
  & ---
  & \textbf{PTAS} \\
Binary $1$DKP
  & FPTAS
  & \textbf{PTAS} ($d=1$)
  & --- \\
Unary knapsack
  & FPTAS
  & \textbf{PTAS}
  & \textbf{PTAS} \\
DAG $s$--$t$ paths
  & exact, $\tfrac{1}{2(1+\eps)}$
  & ---
  & \textbf{PTAS} \\
Independent sets
  & none
  & ---
  & --- \\
Cuts $\delta(U)$
  & GW $\approx 0.878$
  & ---
  & --- \\
$2$DKP
  & max-sum PTAS
  & \textbf{PTAS}; no FPTAS/EPTAS
  & --- \\
Covering $1$DKP
  & min-sum FPTAS
  & min.\ \textbf{PTAS} ($d=1$)
  & --- \\
Covering $2$DKP
  & min-sum PTAS
  & min.\ \textbf{PTAS}; no FPTAS
  & --- \\
\bottomrule
\end{tabular}
\end{table}

On cardinality both PTAS columns apply; the transfer is weaker than $1-1/e$ and weaker than the EPTAS of~\citet{segev2021} for $k=1$.
On cuts only the transfer applies.
Binary knapsack and two-dimensional knapsack sit on the packing ladder and not on exact-sum; DAG paths sit on exact-sum and not on packing.

\section{Conclusion and Open Problems}
\label{sec:concl}

The main line is a four-rung ladder on fixed-dimensional packing, and the covering counterpart for minimization.
There is no single constant-factor algorithm on every membership family; an $\alpha$-approximate max-sum oracle yields ratio $\alpha/((1+\alpha)(1+\eps))$ for every $k$, and every packing family inherits that constant from a deterministic max-sum PTAS.
A constant-factor max-sum oracle is not enough for a PTAS: cuts are a witness (Theorem~\ref{thm:noptas}).
The difference from the $k=1$ maxima literature is the signature, not the surplus scale: void probability is replaced by two approximations to $p\mapsto\E[\min(k,N(p))]$, which then sit on the same packing and exact-sum realization primitives.
Occupancy and quantile signatures, realized in the packing LP, yield a PTAS for every $k\ge 1$ on every fixed-$d$ packing family.
Two-dimensional knapsack lies in the class and rules out a generic FPTAS unless $P=NP$ and a generic EPTAS unless $W[1]=FPT$.
In this setting a generic PTAS is essentially optimal: there is no uniformly better scheme that uses only the packing LP.

The covering side repeats the same ladder with the inequalities reversed.
A min-sum oracle yields a constant factor (Theorem~\ref{thm:min-transfer}); every fixed-$d$ covering family inherits that constant from a deterministic min-sum PTAS.
The same signatures, realized by raising fractionals in the covering LP, yield a PTAS for every $k\ge 1$ (Theorem~\ref{thm:covering}).
Two-dimensional covering knapsack rules out a generic FPTAS (Theorem~\ref{thm:min-limits}).
A min-sum constant on an arbitrary up-set is not enough: vertex cover is the covering analogue of the cut witness.

Query-weight exact-sum is a second PTAS boundary in both directions, incomparable with packing and with covering (DAG paths versus binary knapsack versus covering knapsack).

The map $S\mapsto\E[\Topk(S)]$ is monotone submodular (Appendix~\ref{app:submod}).
Constant-factor adaptivity gaps for stochastic monotone submodular maximization and for stochastic probing~\citep{asadpour2016,gupta2016probing} therefore compare the adaptive and non-adaptive optima by a constant.
The non-adaptive constant-factor schemes of Theorems~\ref{thm:transfer} and~\ref{thm:min-transfer} are, for that reason, also constant-factor approximations of the corresponding adaptive optima; they extend to the adaptive model with no further loss beyond the gap.
The packing and covering PTAS do not transfer in the same way.
An adaptive PTAS for $\E[\Topk]$ under a general packing or covering family---or even a proof that the signature decoder yields one---remains open.
The adaptive ProbeTop-$k$ PTAS of~\citet{fu2018} is for the same sum objective, but only under a cardinality probe budget and only for constant $k$; it does not extend the signature decoder to a general $\F$, nor supply a non-adaptive packing PTAS.

The hardness statements are uniform.
The following per-instance questions remain open.
\begin{itemize}[leftmargin=1.4em,itemsep=2pt]
\item
Which families necessarily admit no constant-factor algorithm for $\E[\Topk]$?
\item
Which families that already have a constant-factor max-sum algorithm necessarily admit no PTAS?
\item
Which exact-sum families, outside packing and covering, admit an EPTAS or an FPTAS? Cardinality for $k=1$ already has an EPTAS~\citep{segev2021,segev2024}.
The same question applies to $\min_{S\in\F}\E[\Topk(S)]$.
\item
Does $\E[\Topk]$ admit an adaptive PTAS on every fixed-$d$ packing family (resp.\ an adaptive minimization PTAS on every fixed-$d$ covering family)?
\end{itemize}

\bibliographystyle{unsrtnat}
\bibliography{refs}

\appendix

\section{Proofs for Section~\ref{sec:prelim}}
\label{app:prelim}

\begin{proof}[Proof of Lemma~\ref{lem:layer}]
Let $x_{(1)}\ge\cdots\ge x_{(m)}$ and $x_{(j)}=0$ for $j>m$, so $\Topk(x)=\sum_{j=1}^k x_{(j)}$.
Each $y\ge 0$ satisfies $y=\int_0^\infty\mathbf{1}_{\{y>t\}}\,dt$.
Tonelli gives
\[
\Topk(x)=\int_0^\infty\#\{j\in[k]:x_{(j)}>t\}\,dt.
\]
The integrand equals $\min(k,\#\{i:x_i>t\})$: if $x_{(k)}>t$ both sides are $k$; if $x_{(k)}\le t$ both sides equal the number of coordinates strictly above $t$.
On a $\delta$-grid the integrand of~\eqref{eq:layer} is constant on each $[r\delta,(r+1)\delta)$.
The variant that counts $\#\{i:x_i\ge r\delta\}$ differs from $\Topk(x)$ by at most $k\delta$.
\end{proof}

\begin{proof}[Proof of Lemma~\ref{lem:pb}]
Induction on the number of summands.
The empty sum is $N=0$.
Write $N'=N+Z$ with $Z\sim B(1,p)$ independent of $N$, and set $A=\E[\min(k,N)]$, $q=\Pr[N<k]$.
Then $\E[\min(k,N')]=A+pq$.
The inductive claim $A(k+\mu)\ge k\mu$ rearranges to the desired
$(A+pq)(k+\mu+p)\ge k(\mu+p)$ as soon as $k-A\le q(k+\mu+p)$.
Here $k-A=\E[(k-N)_+]\le kq$, and $kq\le q(k+\mu+p)$, so the comparison holds.
\end{proof}

\begin{proof}[Proof of Lemma~\ref{lem:f-lip}]
Coupling $N$ and $N^{(-i)}$ by adding an independent Bernoulli of parameter $p_i$ yields
$\partial f/\partial p_i=\Pr[N^{(-i)}<k]\in[0,1]$.
Thus $f$ is nondecreasing, and each partial derivative is itself nonincreasing in every coordinate, so $f$ is concave.
The bound $\lvert f(p)-f(q)\rvert\le\|p-q\|_1$ is the mean-value form of those partials (or, equivalently, a coordinatewise coupling).
Finally $0\le\min(k,N)\le N$ and $\E[N]=\sum_i p_i$ give $0\le f(p)\le\min(k,\sum_i p_i)$.
\end{proof}

\begin{proof}[Proof of Lemma~\ref{lem:pb-chernoff}]
$\mathrm{Var}(N)\le\mu$ gives both tails via the standard multiplicative Chernoff bound (or Bernstein); the upper tail uses the denominator $2(\mu+t/3)$.
For $\mu\ge 2k+C_0\log(k/\eps)$ one has $t=\mu-k$ and $t^2/(2\mu)\ge 3\log(k/\eps)$ after enlarging $C_0$, so $\Pr[N<k]\le\eps^2/k$ and $k-f(p)\le\eps^2$.
\end{proof}

\begin{proof}[Proof of Lemma~\ref{lem:lecam}]
The total-variation bound is Le Cam's theorem.
The function $g$ is $k$-Lipschitz in total variation in the sense that $\lvert\E[g(N+M)]-\E[g(Z+M)]\rvert\le k\,d_{\mathrm{TV}}(N,Z)$.
\end{proof}

\begin{proof}[Proof of Lemma~\ref{lem:pois-mean}]
Coupling Poisson processes, incrementing $\lambda$ by $d\lambda$ adds a point with intensity $d\lambda$, which changes $\min(k,Z_\lambda+M)$ by at most $1$ in expectation per unit of intensity.
\end{proof}

\section{Proofs for Section~\ref{sec:surplus}}
\label{app:surplus}
\label{app:transfer}

\begin{proof}[Proof of Lemma~\ref{lem:det-sand}]
Each $x_i\le T+\Tr(x_i,T)$, so the $k$ largest coordinates sum to at most $kT+\sum_i\Tr(x_i,T)$.
For the lower bound, let $m_+=\#\{i:x_i>T\}$ and $Y=\sum_i\Tr(x_i,T)$.
If $m_+\ge k$ then every one of the $k$ largest exceeds $T$.
If $m_+\le k$ then $\Topk(x)\ge Y+Tm_+\ge Y\ge kT$.
\end{proof}

\begin{proof}[Proof of Lemma~\ref{lem:tail}]
The claim rearranges as $\int_0^T k\mu/(k+\mu)\,dt\ge\int_T^\infty \mu^2/(k+\mu)\,dt$.
Both $u\mapsto ku/(k+u)$ and $u\mapsto u^2/(k+u)$ are nondecreasing.
Set $\alpha_-=\inf_{t<T}\mu(t)\ge\mu(T)$ and $\alpha_+=\sup_{t>T}\mu(t)\le\mu(T)$.
The two sides are at least $T\cdot k\alpha_-/(k+\alpha_-)$ and at most $kT\cdot\alpha_+/(k+\alpha_+)$ respectively, and $\alpha_-\ge\alpha_+$.
\end{proof}

\begin{proof}[Proof of Lemma~\ref{lem:one-set}]
Let $N(t)=\#\{i:X_i>t\}$ and $\mu(t)=\E[N(t)]$.
Independence makes $N(t)$ Poisson-binomial.
Lemmas~\ref{lem:pb} and~\ref{lem:layer} and Tonelli give
\[
\E[\Topk(X)]=\int_0^\infty\E[\min(k,N(t))]\,dt\ge\int_0^\infty\frac{k\mu(t)}{k+\mu(t)}\,dt.
\]
The hypothesis is $\int_T^\infty\mu=kT$, and $\mu$ is nonincreasing, so Lemma~\ref{lem:tail} applies.
If instead $\int_T^\infty\mu\ge kT$, the continuous map $U\mapsto\int_U^\infty\mu-kU$ is nonnegative at $U=T$ and tends to $-\infty$, so it has a root $T'\ge T$; the equality case at $T'$ gives $\E[\Topk]\ge kT'\ge kT$.
\end{proof}

\begin{proof}[Proof of Theorem~\ref{thm:surplus}]
Lemma~\ref{lem:det-sand} yields $\E[\Topk(S)]\le kT+\sum_{i\in S}\E[\Tr(X_i,T)]\le 2kT$ for every $S\in\F$.
Lemma~\ref{lem:one-set} on $S^*$ yields the lower bound.
\end{proof}

\begin{proof}[Proof of Theorem~\ref{thm:transfer}]
Part~(i) is the lottery reduction of Appendix~\ref{app:hardness}, instantiated on independent sets.
For part~(ii), let the loop halt at $T_0$ with output $S_0$.
Then $\sum_{i\in S_0}\E[\Tr(X_i,T_0)]>kT_0$ (or $T_0=0$, in which case the instance is identically zero).
Lemma~\ref{lem:one-set} yields $\E[\Topk(S_0)]\ge kT_0$.

Write $T_1:=T_0(1+\eps)$ for the previous height (if the loop never ran, every surplus at $T_0$ is already larger than $kT_0$ and the same comparison applies with $T_1=T_0$).
At $T_1$ the while-condition still held, so the oracle set $S_1$ satisfied $\E[\Tr(S_1,T_1)]\le kT_1$.
Because $\mathcal{A}$ is an $\alpha$-approximation,
\[
\max_{S\in\F}\sum_{i\in S}\E[\Tr(X_i,T_1)]
\le \frac{kT_1}{\alpha}.
\]
Lemma~\ref{lem:det-sand} and taking expectations give, for every $S\in\F$,
\[
\E[\Topk(S)]
\le kT_1+\sum_{i\in S}\E[\Tr(X_i,T_1)]
\le kT_1\Bigl(1+\frac{1}{\alpha}\Bigr)
=kT_0(1+\eps)\frac{1+\alpha}{\alpha}.
\]
Dividing the lower bound $kT_0$ by this upper bound produces the stated ratio.
The number of iterations is $O(\eps^{-1}\log(\mathrm{range}))$ as in the $k=1$ search: after that many geometric steps $T$ is smaller than every positive truncated mean that can occur, and the instance is treated as identically zero.
\end{proof}

\begin{proof}[Proof of Theorem~\ref{thm:min-transfer}]
If every law is identically zero the claim is vacuous.
Let the loop halt at $T_0$ with output $S_0$.
Then $\sum_{i\in S_0}\E[\Tr(X_i,T_0)]<\alpha kT_0$ (or $T_0$ exceeds every atom, in which case the surplus is zero).
Lemma~\ref{lem:det-sand} yields
\[
\E[\Topk(S_0)]
\le
kT_0+\sum_{i\in S_0}\E[\Tr(X_i,T_0)]
<
(1+\alpha)kT_0.
\]
Write $T_1:=T_0/(1+\eps)$ for the previous height (if the loop never ran, take $T_1=T_0$).
At $T_1$ the while-condition held, so the oracle set $S_1$ satisfied $\E[\Tr(S_1,T_1)]\ge\alpha kT_1$.
If $\OPT_{\min}<kT_1$, an optimum $S^*$ would satisfy $\E[\Topk(S^*)]<kT_1$, hence $\sum_{i\in S^*}\E[\Tr(X_i,T_1)]<kT_1$ by the contrapositive of Lemma~\ref{lem:one-set}.
The true min-sum at height $T_1$ would then be less than $kT_1$, and $\mathcal{A}$ would return a set with surplus less than $\alpha kT_1$, a contradiction.
Thus $\OPT_{\min}\ge kT_1=kT_0/(1+\eps)$.
Dividing the upper bound $(1+\alpha)kT_0$ by this lower bound produces the stated ratio.
The number of iterations is $O(\eps^{-1}\log(\mathrm{range}))$ as in Algorithm~\ref{alg:transfer}.
\end{proof}

\section{Proof of Proposition~\ref{prop:es-to-ms}}
\label{app:es-to-ms}

\begin{proof}[Proof of Proposition~\ref{prop:es-to-ms}]
Drop every ground-set element that lies in no member of $\F$.
Let $W:=\max_i w_i$.
Nonnegativity yields $\OPT_w\ge W$: the heaviest remaining element belongs to some feasible set.
If $W=0$ return any nonempty feasible set (or $\emptyset$).
Otherwise set $\delta:=\eps W/n$ and $w_i':=\bigl\lfloor w_i/\delta\bigr\rfloor$.
Then $0\le w_i'\le n/\eps$ and $0\le w_i-\delta w_i'<\delta$.
For every $S\subseteq[n]$,
\[
0\le w(S)-\delta\,w'(S)<n\delta=\eps W\le\eps\,\OPT_w.
\]
Query-weight exact-sum on $\{w_i'\}$ decides, for each target $t\in\{0,1,\dots,n\lfloor n/\eps\rfloor\}$, whether some $S\in\F$ has $w'(S)=t$, in time $\poly(\langle\F\rangle)\,(n/\eps)^{O(1)}$ per call.
Trying all $t$ (or scanning downward from the largest feasible $t$) returns a set $S'$ maximizing $w'$.
If $S^*$ is a max-sum optimum of the original weights,
\[
w(S')\ge\delta\,w'(S')\ge\delta\,w'(S^*)>w(S^*)-\eps\,\OPT_w=(1-\eps)\,\OPT_w.
\]
The number of calls is $O(n^2/\eps)$, so the total time is $\poly(\langle\F\rangle,n,1/\eps)$.
This is an FPTAS for max-sum, hence a constant-factor approximation (e.g.\ $\eps=1/2$).
Theorem~\ref{thm:transfer}(ii) then applies.
Scanning the smallest feasible $t$ instead of the largest is a min-sum FPTAS on the same rounded instance, so Theorem~\ref{thm:min-transfer} likewise applies on every query-weight family.

The original magnitude of $w$ may be $2^{\Theta(L)}$ for bit length $L$.
The procedure never enumerates targets on that scale: it enumerates targets only after rounding.
Exact max-sum on the original $w$ need not lie in $\mathrm{P}$.
\end{proof}

\section{Proofs for Section~\ref{sec:disc}}
\label{app:disc}

\begin{proof}[Proof of Lemma~\ref{lem:tailmass}]
Pathwise $\Topk\ge M\cdot\min(k,N(M))$, so $\E[\min(k,N(M))]\le cW/M=c\eps$.
Lemma~\ref{lem:pb} yields $k\mu(M)/(k+\mu(M))\le c\eps$.
The hypothesis $c\eps<k/2$ rearranges to $\mu(M)\le c\eps\cdot k/(k-c\eps)\le 2c\eps$.
For $t\ge M$ one has $\mu(t)\le\mu(M)\le 2c\eps$, and
\[
\Pr[\max_i X_i>t]=1-\prod_i(1-p_i(t))\ge\mu(t)-\mu(t)^2\ge\mu(t)\bigl(1-2c\eps\bigr).
\]
Hence $\E[\Tr(\max_i X_i,M)]\ge(1-2c\eps)\Sigma$.
The left side is at most $\E[\Topk]\le cW$, so $\Sigma\le cW/(1-2c\eps)\le 2cW$ once $c\eps\le 1/4$.
\end{proof}

\begin{proof}[Proof of Lemma~\ref{lem:fold}]
Write $q_i:=\tau_i/M$ and $d_i(t):=\bar p_i(t)-p_i(t)\ge 0$ for $t<M$, and split
\begin{align*}
\E[\Topk(X)]
&=\int_0^M\E[\min(k,N)]\,dt+\int_M^\infty\E[\min(k,N)]\,dt,\\
\E[\Topk(\bar X)]
&=\int_0^M\E[\min(k,\bar N)]\,dt.
\end{align*}
Lemma~\ref{lem:tailmass} gives $\Sigma:=\sum_i\tau_i\le 2cW$ and $q_i\le 2c\eps$.
For $t\ge M$, Lemma~\ref{lem:pb} and $\mu(t)\le 2c\eps$ give
$\E[\min(k,N(t))]\ge\mu(t)\,(1-2c\eps/k)$,
so the tail integral lies in $[(1-2c\eps/k)\Sigma,\Sigma]$.

Lemma~\ref{lem:f-lip} yields $\Delta(t):=\E[\min(k,\bar N(t))]-\E[\min(k,N(t))]\le\sum_i d_i(t)\le\sum_i q_i=\Sigma/M$.
Hence $\int_0^M\Delta\le\Sigma$ and
\[
\E[\Topk(\bar X)]
\le
\E[\Topk(X)]-\int_M^\infty\E[\min(k,N)]+\Sigma
\le
\E[\Topk(X)]+O\bigl((c^2\eps)/k\bigr)\,W.
\]

For the matching lower bound, concavity of $f$ along the segment from $p(t)$ to $\bar p(t)$ gives
$\Delta(t)\ge\sum_i d_i(t)\,\Pr[\bar N^{(-i)}(t)<k]\ge\bigl(\sum_i d_i(t)\bigr)\Pr[\bar N(t)<k]$.
If the cap in Definition~\ref{def:fold} does not bind then $\sum_i d_i(t)=\Sigma/M$ constantly.
For each $i$ write $C_i:=\{t\in[0,M):p_i(t)>1-q_i\}$.
On $C_i$ one has $p_i(t)>1-2c\eps$, hence $\E[\min(k,N(t))]\ge 1-O(\eps)$.
Integrating against $\E[\Topk]\le cW$ therefore gives $\lvert C_i\rvert=O(\eps M)$ (using $M=W/\eps$).
When the cap binds, $d_i(t)=1-p_i(t)<q_i$, so the mass lost to the cap is
\[
\sum_i\int_{C_i}\bigl(q_i-d_i(t)\bigr)\,dt
\le
\sum_i q_i\lvert C_i\rvert
=
O(\eps)\,M\sum_i q_i
=
O(\eps)\,\Sigma.
\]
Thus $\int_0^M\sum_i d_i=\Sigma-O(\eps)\Sigma$.
Since $\Pr[\bar N<k]\le 1$,
\[
\int_0^M\Delta
\ge
\int_0^M\Pr[\bar N<k]\sum_i d_i\,dt
\ge
\frac{\Sigma}{M}\int_0^M\Pr[\bar N<k]\,dt-O(\eps)\Sigma.
\]
As in the $k$-th order statistic identity,
$\int_0^M\Pr[\bar N\ge k]\,dt=\E[\min(M,X_{(k)}(\bar X))]\le\E[\Topk(\bar X)]\le cW+O((c^2\eps)/k)W$,
so $\int_0^M\Pr[\bar N<k]\,dt\ge M\bigl(1-O(c\eps)\bigr)$.
Therefore $\int_0^M\Delta\ge(1-O(c\eps))\Sigma$ and
\[
\E[\Topk(\bar X)]
\ge
\E[\Topk(X)]-\Sigma+(1-O(c\eps))\Sigma
\ge
\E[\Topk(X)]-O(c^2\eps)\,W.
\]
\end{proof}

\begin{proof}[Proof of Lemma~\ref{lem:fold-back}]
Write $I_{\mathrm{low}}:=\int_0^M\E[\min(k,N)]\,dt$, $\Sigma:=\sum_i\tau_i$, and $U:=\E[\Topk(\bar X)]$.
If some $\tau_i\ge M$ then $\bar p_i\equiv 1$ on $[0,M)$, so $U\ge M=W/\eps>(1+4\eps)W$, contradicting $U\le(1+4\eps)W$.
Thus $q_i:=\tau_i/M<1$ for every $i$.
Stochastic domination on $[0,M]$ gives $I_{\mathrm{low}}\le U$.
On $C_i:=\{t\in[0,M):p_i(t)>1-q_i\}$ one has $\bar p_i=1$, hence $\E[\min(k,\bar N(t))]\ge 1$ and $\lvert C_i\rvert\le U\le(1+4\eps)\eps M$.
The mass lost to the cap is at most
\[
\sum_i q_i\lvert C_i\rvert
\le
(1+4\eps)\eps\,\Sigma.
\]
Write $\gamma:=(1+4\eps)\eps$ and $d_i=\bar p_i-p_i$.
Then $\int_0^M\bigl(\Sigma/M-\sum_i d_i\bigr)\,dt\le\gamma\Sigma$, so
\[
\int_0^M\Pr[\bar N<k]\sum_i d_i\,dt
\ge
\frac{\Sigma}{M}\int_0^M\Pr[\bar N<k]\,dt-\gamma\Sigma.
\]
The identity $\int_0^M\Pr[\bar N\ge k]=\E[\min(M,X_{(k)}(\bar X))]\le U$ yields
$\int_0^M\Pr[\bar N<k]\,dt\ge M-U$, and therefore
\[
U
=\E[\Topk(\bar X)]
\ge
I_{\mathrm{low}}+\Sigma\Bigl(1-\frac{U}{M}-\gamma\Bigr)
\ge
I_{\mathrm{low}}+(1-4\eps)\Sigma,
\]
where the last step uses $U/M+\gamma\le 2(1+4\eps)\eps\le 4\eps$ for $\eps\le 1/20$.
Hence $\Sigma\le U/(1-4\eps)$.
Pathwise $\E[\Topk(X)]=I_{\mathrm{low}}+I_{\mathrm{high}}$ with $I_{\mathrm{high}}\le\Sigma$, so
\[
\E[\Topk(X)]
\le
I_{\mathrm{low}}+\Sigma
\le
U+4\eps\,\Sigma
\le
\frac{U}{1-4\eps}.
\]
For $\eps\le 1/20$ one has $4\eps/(1-4\eps)\le 5\eps$, hence $U/(1-4\eps)\le(1+5\eps)U$.
\end{proof}

\begin{proof}[Proof of Corollary~\ref{cor:min-allk}]
Proposition~\ref{prop:es-to-ms} supplies a min-sum FPTAS.
Run Algorithm~\ref{alg:min-transfer} with $\alpha=1+\eps$ and set $W:=\E[\Topk(S_0)]$, so $\OPT_{\min}\le W\le(1+\eps)(2+\eps)\,\OPT_{\min}$.
Fold and grid at this $W$, after replacing the internal accuracy by $\Theta(\eps)$ small enough that Theorem~\ref{thm:disc} with $c=1$ costs at most $\eps W$ (the substitution of Corollary~\ref{cor:ptas} absorbs the $k$ in the grid).
A min-optimum $S^*$ satisfies $\E[\Topk(S^*)]\le W$, so that application is legal.
Run Algorithm~\ref{alg:ptas} or~\ref{alg:quantile} and return the realized set of smallest folded $\E[\Topk]$.
Let $\bar S$ be optimal for the folded laws and let $\tilde S$ be the enumerated set with the same signature.
The decoder matches signatures, so
\[
\E[\Topk(\tilde S;\bar X)]
\le
\E[\Topk(\bar S;\bar X)]+\eps W
\le
\E[\Topk(S^*;\bar X)]+\eps W
\le
\OPT_{\min}+2\eps W.
\]
The right-hand side is at most $(1+2\eps(1+\eps)(2+\eps))\,\OPT_{\min}\le(1+4\eps)W$ for $\eps\le 1/20$.
Lemma~\ref{lem:fold-back} with this $U\le(1+4\eps)W$ yields
\[
\E[\Topk(\tilde S;X)]
\le
\frac{\OPT_{\min}+2\eps W}{1-4\eps}
\le
(1+O(\eps))\,\OPT_{\min}.
\]
(The same bound holds for $\bar S$.)
\end{proof}

\begin{proof}[Proof of Lemma~\ref{lem:grid}]
Pathwise $0\le X_i-\tilde X_i<\delta$, so the $k$ largest drop by less than $k\delta$ in total.
The number of grid points on $[0,M]$ is $M/\delta=1/\eps^2$.
\end{proof}

\begin{proof}[Proof of Theorem~\ref{thm:disc}]
Combine Lemmas~\ref{lem:fold} and~\ref{lem:grid}.
\end{proof}

\section{Proofs for Section~\ref{sec:sig}}
\label{app:sig}

\begin{proof}[Proof of Lemma~\ref{lem:layer-unsat}]
Let $\mu_t$ be the tiny mass.
Quantization gives $\lvert\mu_t-\gamma\sigma\rvert\le n\gamma=\eps^2$.
Lemma~\ref{lem:lecam} with $M=N(p^{\mathrm{big}})$ costs $k\sum_{\mathrm{tiny}}p_i^2\le k\eta\mu_t\le k\eta\Lambda=\eps^2$.
Lemma~\ref{lem:pois-mean} moves $\mu_t$ to $\gamma\sigma$ at cost $\eps^2$.
Lemma~\ref{lem:f-lip} moves $p^{\mathrm{big}}$ to the bucket representatives at cost at most $(\Lambda/\eta)\beta=\eps^2/2$.
The four errors sum to at most $4\eps^2$, and the resulting law is $\widehat N$.
\end{proof}

\begin{proof}[Proof of Lemma~\ref{lem:layer-sat}]
If $\mu\ge\Lambda$, Lemma~\ref{lem:pb-chernoff} applies.
If $\mu<\Lambda$, then $n_{\mathrm{big}}\le\mu/\eta\le\Lambda/\eta=k\Lambda^2/\eps^2$, so the proxy mean differs from $\mu$ by at most $n_{\mathrm{big}}\beta+n\gamma\le\eps^2/2+\eps^2$.
Saturation forces the proxy mean $\ge\Lambda$, hence $\mu\ge\Lambda-O(\eps^2)$.
Enlarge $C_0$ in Lemma~\ref{lem:pb-chernoff} by a constant so that $\mu\ge\Lambda-O(\eps^2)$ still yields $k-f(p)\le O(\eps^2)$.
\end{proof}

\begin{proof}[Proof of Theorem~\ref{thm:sig}]
The two sets share $(\sigma,(n_b))$ and the saturation bit at every layer, hence share $\widehat N_r$.
Lemmas~\ref{lem:layer-unsat} and~\ref{lem:layer-sat} give an $O(\eps^2)$ gap to the proxy on each layer.
Multiply by $\delta=\eps W$ and sum $h=1/\eps^2$ layers to obtain an $O(\eps)\,W$ gap, with no further factor of $k$.
\end{proof}

\section{Monotone submodularity}
\label{app:submod}

\begin{lemma}
\label{lem:submod}
The set function $S\mapsto\E[\Topk(S)]$ is nonnegative, monotone, and submodular.
\end{lemma}

\begin{proof}
Nonnegativity is immediate.
Monotonicity: enlarging $S$ can only increase the $k$ largest coordinates, pathwise.
Submodularity: it is enough to check the pathwise function $S\mapsto\Topk((x_i)_{i\in S})$ for a fixed $x\in\R_+^n$, then take expectations.
Let $A\subseteq B$ and $e\notin B$.
Write $\Delta(S)=\Topk(S\cup\{e\})-\Topk(S)$.
If $x_e$ does not enter the top $k$ of $S\cup\{e\}$, then $\Delta(S)=0$.
Otherwise $x_e$ replaces the current $k$-th order statistic of $S$, so $\Delta(S)=(x_e-x_{(k)}(S))_+$.
The $k$-th order statistic is nondecreasing in the set, hence $\Delta(A)\ge\Delta(B)$.
\end{proof}

The greedy algorithm for monotone submodular maximization under cardinality therefore yields a $(1-1/e)$-approximation~\citep{nemhauser1978}, and the continuous greedy algorithm extends the same ratio to matroids~\citep{calinescu2011}.
That baseline is independent of $k$ and of the laws, but it is not a PTAS.

\section{Proofs for Section~\ref{sec:quantile}}
\label{app:quantile}

\begin{proof}[Proof of Lemma~\ref{lem:path-proxy}]
Lemma~\ref{lem:det-sand} gives $A_t(x)\ge\Topk(x)$ for every $t\ge 0$.
Let $x_{(1)}\ge\cdots\ge x_{(m)}$ and set $t^*=x_{(k)}$ (value $0$ if $m<k$).
Then $N(x,t^*)=\#\{i:x_i>t^*\}\le k$ and $N'(x,t^*)=\#\{i:x_i\ge t^*\}\ge k$, so the $k$ largest coordinates consist of every value strictly above $t^*$ together with $k-N(x,t^*)$ copies of $t^*$.
Hence
\[
\Topk(x)
=\sum_{x_i>t^*}x_i+(k-N(x,t^*))t^*
=k t^*+\sum_i\Tr(x_i,t^*)
=A_{t^*}(x),
\]
and the minimum is attained.

Now fix an arbitrary $T$ and write $H=N(x,T)$, $M=N'(x,T)$, $z_i=\Tr(x_i,T)$.
If $H\ge k$ then $M\ge H\ge k$, so $(k-M)_+=0$, the $k$ largest coordinates all exceed $T$, and $A_T-\Topk$ equals the sum of the $H-k$ smallest positive $z_i$, which is at most $\frac{H-k}{H}\sum_i z_i$.
If $M<k$ then $H\le M<k$, so $(H-k)_+=0$ and every coordinate at least $T$ enters the top $k$; the remaining $k-M$ selected coordinates are at most $T$, hence $A_T-\Topk\le(k-M)T$.
If $H\le k\le M$ both terms vanish and the $k$ largest coordinates are exactly the coordinates above $T$ together with $k-H$ copies of $T$, so $A_T=\Topk$ as at $t^*$.
\end{proof}

\begin{proof}[Proof of Lemma~\ref{lem:c1}]
Write $H=N(X,T)$, $M=N'(X,T)$, $\mu=\E[H]$, $\mu'=\E[M]$, $Z_i=\Tr(X_i,T)$, and $V=\Topk(X)$.
The hypothesis is $\mu\le k\le\mu'$ (the slackened window is treated at the end).
Pathwise $A_T\ge V$, so the left-hand side of~\eqref{eq:c1} is $\E[A_T-V]$.
Lemma~\ref{lem:path-proxy} gives
\[
A_T-V
\le
T(k-M)_+
+\frac{(H-k)_+}{H}\sum_i Z_i.
\]

For the shortage, $\mu'\ge k$.
The coordinates $B'_i=\mathbf{1}_{\{X_i\ge T\}}$ are independent Bernoullis.
Let $q:=k/\mu'\in(0,1]$ and let $\xi_i$ be independent $\Bern(q)$, independent of $(X_j)_j$.
The thinned sum $Y:=\sum_i B'_i\xi_i$ is Poisson-binomial with mean $k$, and $Y\le M$ almost surely, so $(k-M)_+\le(k-Y)_+$.
Hence
\[
\E[(k-M)_+]
\le
\E\lvert Y-k\rvert
\le
\sqrt{\mathrm{Var}(Y)}
\le\sqrt{k},
\]
and $\E[T(k-M)_+]\le kT/\sqrt{k}$.

For the overflow write $B_i=\mathbf{1}_{\{X_i>T\}}$ and $H_{-i}=\sum_{j\neq i}B_j$.
The identity $Z_i>0\Rightarrow B_i=1$ gives
\[
Z_i\,\frac{(H-k)_+}{H}
=
Z_i\,\frac{(H_{-i}+1-k)_+}{H_{-i}+1}
\]
(the two sides vanish when $H=0$).
Independence of $Z_i$ and $H_{-i}$ therefore produces
\begin{equation}
\label{eq:loo}
\E\Bigl[\frac{(H-k)_+}{H}\sum_i Z_i\Bigr]
=
\sum_i\E[Z_i]\,\E\Bigl[\frac{(H_{-i}+1-k)_+}{H_{-i}+1}\Bigr].
\end{equation}
On $\{H_{-i}+1>k\}$ one has $1/(H_{-i}+1)\le 1/k$, so the inner expectation is at most $k^{-1}\E[(H_{-i}+1-k)_+]$.
Here $\E[H_{-i}]\le\mu\le k$ and $\mathrm{Var}(H_{-i})\le k$, hence
\[
\E[(H_{-i}+1-k)_+]
\le
\sqrt{\mathrm{Var}(H_{-i})}+(\E[H_{-i}]+1-k)_+
\le
\sqrt{k}+1.
\]
Substituting into~\eqref{eq:loo},
\[
\E\Bigl[\frac{(H-k)_+}{H}\sum_i Z_i\Bigr]
\le
\bigl(k^{-1/2}+k^{-1}\bigr)\sum_i\E[Z_i].
\]
Adding the shortage bound and using $A_T=kT+\sum_i Z_i$ yields~\eqref{eq:c1}.
A high-probability window of width $\sqrt{k\log m}$ is not used: $\E\lvert Y-k\rvert=O(\sqrt{k})$ already controls the shortage, and~\eqref{eq:loo} transfers the same $O(\sqrt{k})$ overflow to the unbounded $Z_i$.

If instead $\mu\le k+\delta$ and $\mu'\ge k-\delta$, thin $M$ to mean $k':=\max\{k-\delta,0\}$.
Then $\E[(k-M)_+]\le\delta+\sqrt{k'}$, which contributes an extra $T\delta\le(\delta/k)\,kT$ relative to $A_T$.
The overflow uses $\mathrm{Var}(H_{-i})\le k+\delta$ and $(\E[H_{-i}]+1-k)_+\le\delta+1$, so
$\E[(H_{-i}+1-k)_+]\le\sqrt{k+\delta}+\delta+1$, which adds $O((\delta+1)/k)$ to $\eta_k$.
\end{proof}

\begin{proof}[Proof of Lemma~\ref{lem:mult-round}]
Write $X'_i=X_i\mathbf{1}_{\{X_i\ge\tau\}}$.
Zeroing every atom below $\tau$ drops each of the $k$ largest coordinates by less than $\tau$, so $0\le\Topk(X)-\Topk(X')<k\tau\le\eps W$.
On $\{X_i\ge\tau\}$ the rounded value $Y_i$ satisfies $Y_i\le X_i<(1+\eps)Y_i$, hence $Y_i>X_i/(1+\eps)$.
Therefore $\Topk(Y)\ge\Topk(X')/(1+\eps)$ and $\Topk(Y)\le\Topk(X)$.
The surviving atoms lie in $[\tau,W/\eps]$, so the number of grid points is
\[
1+\log_{1+\eps}\frac{W/\eps}{\tau}
=
O\bigl(\eps^{-1}\log(n/\eps)\bigr).
\]
\end{proof}

\section{Packing LP realization}
\label{app:packing}

\begin{lemma}[Dropping fractionals]
\label{lem:lp-drop}
Let $\F$ be a $d$-dimensional packing family, let $\sigma_i\in\R_+^D$, and let $\tau_i\in\R_+^C$.
Suppose the polytope
\begin{equation}
\label{eq:sig-lp}
\begin{aligned}
Ax&\le b,\\
\lvert (\sigma^\top x)_\ell-z_\ell\rvert&\le\delta_\ell
&&(\ell=1,\dots,D),\\
(\tau^\top x)_c&\ge\gamma_c
&&(c=1,\dots,C),\\
0&\le x\le 1
\end{aligned}
\end{equation}
is nonempty.
A basic feasible solution has at most $F:=d+2D+C$ strictly fractional coordinates.
Zeroing those coordinates yields an integral $S\in\F$.
Writing $\Sigma(S)$ and $\Gamma(S)$ for the realized $\sigma$- and $\tau$-sums,
\[
\bigl\lvert\Sigma_\ell(S)-z_\ell\bigr\rvert
\le
\delta_\ell+F\max_{i:\,0<x_i<1}\sigma_{i\ell},
\qquad
\Gamma_c(S)
\ge
\gamma_c-F\max_{i:\,0<x_i<1}\tau_{ic}.
\]
\end{lemma}

\begin{proof}
A nonempty polytope defined by $d+2D+C$ non-box inequalities and box constraints has a BFS, and a BFS has at most that many strictly fractional coordinates.
Packing inequalities have nonnegative coefficients, so setting a coordinate to $0$ preserves $Ax\le b$ and membership in $\F$.
Each dropped item changes a $\sigma$- or $\tau$-coordinate by at most the corresponding entry of that item.
Covering inequalities need not be preserved exactly: the displayed lower bound records the worst-case loss.
\end{proof}

\begin{lemma}[Raising fractionals]
\label{lem:lp-lift}
Let $\F$ be a $d$-dimensional covering family, let $\sigma_i\in\R_+^D$, and let $\tau_i\in\R_+^C$.
Suppose the polytope
\begin{equation}
\label{eq:sig-lp-cov}
\begin{aligned}
Ax&\ge b,\\
\lvert (\sigma^\top x)_\ell-z_\ell\rvert&\le\delta_\ell
&&(\ell=1,\dots,D),\\
(\tau^\top x)_c&\ge\gamma_c
&&(c=1,\dots,C),\\
0&\le x\le 1
\end{aligned}
\end{equation}
is nonempty.
A basic feasible solution has at most $F:=d+2D+C$ strictly fractional coordinates.
Raising those coordinates to $1$ yields an integral $S\in\F$.
Writing $\Sigma(S)$ and $\Gamma(S)$ for the realized $\sigma$- and $\tau$-sums,
\[
\bigl\lvert\Sigma_\ell(S)-z_\ell\bigr\rvert
\le
\delta_\ell+F\max_{i:\,0<x_i<1}\sigma_{i\ell},
\qquad
\Gamma_c(S)
\ge
\gamma_c.
\]
\end{lemma}

\begin{proof}
A BFS has at most $d+2D+C$ strictly fractional coordinates, as in Lemma~\ref{lem:lp-drop}.
Covering inequalities have nonnegative coefficients, so setting a coordinate to $1$ preserves $Ax\ge b$ and membership in $\F$.
Each raised item changes a $\sigma$-coordinate by at most the corresponding entry; every covering signature inequality can only improve.
\end{proof}

\begin{lemma}[Occupancy robustness]
\label{lem:occ-robust}
Assume the fold and grid of Theorem~\ref{thm:disc}.
Let $H\subseteq[n]$ and let $R_+\subseteq\{0,\dots,h-1\}$.
Let $S=H\cup R$ and $S'=H\cup R'$ be disjoint unions, and write $\mu_r(\,\cdot\,)$ for layer-$r$ occupancy means.
Suppose that every residual coordinate on every layer $r\notin R_+$ is at most $\eta_{\mathrm{LP}}$, that $\lvert\mu_r(R)-\mu_r(R')\rvert\le 4\eps^2$ for every $r\notin R_+$, and that every $r\in R_+$ has $\mu_r(S),\mu_r(S')\ge\Lambda$.
Then
\[
\bigl\lvert\E[\Topk(S)]-\E[\Topk(S')]\bigr\rvert\le O(\eps)\,W.
\]
\end{lemma}

\begin{proof}
Fix a layer $r\notin R_+$.
Write $N_H$ for the occupancy of $H$ and $N_R$, $N_{R'}$ for the residuals.
The two full occupancies are $N_H+N_R$ and $N_H+N_{R'}$, with $N_H$ common.
Let $\mu=\mu_r(R)$ and $\mu'=\mu_r(R')$.
Lemma~\ref{lem:lecam} gives
\[
\bigl\lvert\E[\min(k,N_H+N_R)]-\E[\min(k,N_H+Z)]\bigr\rvert
\le
k\sum_{i\in R}p_{i,r}^2
\le
k\eta_{\mathrm{LP}}\mu,
\]
where $Z\sim\mathrm{Poisson}(\mu)$ is independent of $N_H$, and likewise for $R'$ at cost $k\eta_{\mathrm{LP}}\mu'$.
Lemma~\ref{lem:pois-mean} moves the Poisson mean from $\mu$ to $\mu'$ at cost $\lvert\mu-\mu'\rvert\le 4\eps^2$.
Unsaturated residual means satisfy $\mu,\mu'\le\Lambda_++O(\eps^2)$, and the choice of $\eta_{\mathrm{LP}}$ yields $k\eta_{\mathrm{LP}}\Lambda_+\le\eps^2/2$.
Hence $\lvert f(p_{\cdot,r}(S))-f(p_{\cdot,r}(S'))\rvert\le O(\eps^2)$.
(The aligned $\ell_1$ identity $\|p(S)-p(S')\|_1=\lvert\mu-\mu'\rvert$ is false when $R$ and $R'$ contain different items; Poissonization does not require a coupling of items.)
On a layer $r\in R_+$ both total means are at least $\Lambda$, so $f\ge k-O(\eps^2)$ in both sets by Lemma~\ref{lem:layer-sat}.
Summing $\delta\sum_r$ with $\delta h=W/\eps$ produces $O(\eps)\,W$.
\end{proof}

\begin{proof}[Proof of Theorem~\ref{thm:packing}]
Run Algorithm~\ref{alg:transfer} with a $(1/2)$-approximate max-sum oracle given by a packing PTAS~\citep{frieze1984,ibarra1975}, and set $W:=kT_0$ at the halt.
Theorem~\ref{thm:transfer}(ii) yields $W\le\OPT\le 5W$.
Work thereafter with the folded laws at this $W$.
Throughout this proof, $\eps$ is the \emph{internal} accuracy of the fold, the grids, and the LP; write $\eps_{\mathrm{out}}$ for the requested PTAS accuracy.
On the large-$k$ range one may take $\eps=\Theta(\eps_{\mathrm{out}})$.
On the small-$k$ range Theorem~\ref{thm:disc} is run at this internal $\eps$ (with an absolute $c=O(1)$); the substitution $\eps\leftarrow\Theta(\eps_{\mathrm{out}}/k)$ of Corollary~\ref{cor:ptas} is applied only at the end of that case.
Let $S^*$ be optimal for the preprocessed instance.
The algorithm enumerates every guess below, realizes an integral packing from each feasible residual LP, and returns the set of largest true $\E[\Topk]$; it is enough to exhibit one guess whose output is an $O(\eps)\,W$-approximation of $S^*$.

\paragraph{Case $k\ge C/\eps^2$.}
Run the preprocessing of Algorithm~\ref{alg:quantile} (with the packing scale $W\le\OPT\le 5W$ in place of the exact-sum scale $2W$) and let $T^*$ be a mixture quantile of $S^*$ on the grid $G\cup\{0\}$.
Set $\eta:=\eps$ and call $i$ surplus-heavy if $\E[\Tr(X_i,T^*)]>\eta W$.
Lemma~\ref{lem:c1} and $\E[\Topk(S^*)]\le 5W$ give $\E[A_{T^*}(S^*)]=O(W)$, hence $\Phi_{S^*}(T^*)=O(W)$.
Thus $S^*$ contains at most $\Phi_{S^*}(T^*)/(\eta W)+1=O(1/\eps)$ surplus-heavies.
Occupancy coordinates of any single item are at most $1$; one does not enumerate occupancy-heavies.
Enumerate every subset $H$ of surplus-heavies of size $O(1/\eps)$ in time $n^{O(1/\eps)}$, and guess a residual target $z$ for $(\mu^>,\mu^\ge,\Phi)$ of $S^*\setminus H$ on the same $\eps$-grid as Algorithm~\ref{alg:quantile} ($O(\poly(n,1/\eps))$ targets).
If $H$ is infeasible skip it.
The residual LP is~\eqref{eq:sig-lp} on $[n]\setminus H$ with residual capacities $b-a(H)$, $D=3$, $C=0$, and $\delta_\ell$ the quantization of Theorem~\ref{thm:quantile} ($O(\eps)$ in occupancy, $O(\eps)\,W$ in surplus).
The indicator of $S^*\setminus H$ is feasible for the correct guess, so a BFS exists.
Lemma~\ref{lem:lp-drop} drops at most $d+6$ fractionals.
Occupancy of the integral residual changes by $O(d)=O(1)$, which is a window slack $\delta=O(1)$ in Lemma~\ref{lem:c1}; then $\eta_k+O((\delta+1)/k)=O(\eps)$ because $k=\Omega(1/\eps^2)$.
Surplus of the residual changes by at most $(d+6)\eta W=O(\eps)\,W$, and $H$ is common to $S^*$ and the output.
The comparison
\[
\E[\Topk(H\cup S)]
\ge
(1-O(\eps))\,\E[A_{T^*}(H\cup S)]
\ge
(1-O(\eps))\,\E[\Topk(S^*)]-O(\eps)\,W
\]
is Lemma~\ref{lem:c1} with occupancy slack $O(1)$ and surplus slack $O(\eps)\,W$.
Taking $\eps=\Theta(\eps_{\mathrm{out}})$ yields an $O(\eps_{\mathrm{out}})\,\OPT$ guarantee.

\paragraph{Case $k\le C/\eps^2$.}
Let $h=1/\eps^2$ and $\Lambda=2k+C_0\log(k/\eps)$ as in Section~\ref{sec:sig}.
Write $F_{\max}:=d+2h$, $\Lambda_+:=\Lambda+F_{\max}$, and
\[
\eta_{\mathrm{LP}}
:=
\frac{\eps^2}{2(k\Lambda_++F_{\max})}.
\]
Guess the set $R_+$ of layers at which $S^*$ has mean at least $\Lambda_+$ ($2^{h}$ guesses).
Call $i$ LP-heavy if $\max_{r\notin R_+}p_{i,r}>\eta_{\mathrm{LP}}$.
This predicate does not depend on $S^*$.
On each layer $r\notin R_+$ one has $\mu_r(S^*)<\Lambda_+$, so $S^*$ contains at most $\Lambda_+/\eta_{\mathrm{LP}}$ LP-heavies on that layer, and at most $h\Lambda_+/\eta_{\mathrm{LP}}$ LP-heavies in total.
Enumerate every set $H$ of LP-heavies of that cardinality in time $n^{O(h\Lambda_+(k\Lambda_++F_{\max})/\eps^2)}$.
For $k=O(1/\eps^2)$ this is $n^{f(d,1/\eps)}$.
If $H$ is infeasible skip it.
The residual universe is the set of non-LP-heavies; every remaining item satisfies $p_{i,r}\le\eta_{\mathrm{LP}}$ for all $r\notin R_+$.

Let $L:=h-|R_+|$ and write $t_r^*:=\mu_r(S^*\setminus H)$ for $r\notin R_+$.
Then $0\le t_r^*<\Lambda_+$.
Guess residual tiny masses $t_r$ in steps of $\eps^2$ on $[0,\Lambda_+]$ ($\bigl(O(\Lambda_+/\eps^2)\bigr)^{O(h)}$ guesses) so that $\lvert t_r-t_r^*\rvert\le\eps^2$ for the correct guess.
The residual LP on the non-LP-heavies, with capacities $b-a(H)$, is~\eqref{eq:sig-lp} with
$D=L$ tiny-mass coordinates, $\delta_r=\eps^2$, $z_r=t_r$, and covering inequalities
\[
\sum_i p_{i,r}x_i
\ge
\bigl(\Lambda_+-\mu_r(H)\bigr)_+
\qquad(r\in R_+),
\]
so $C=|R_+|$.
The number of non-box constraints is
\[
d+2L+|R_+|
=
d+2\bigl(h-|R_+|\bigr)+|R_+|
=
d+2h-|R_+|
\le
F_{\max}.
\]
The indicator of $S^*\setminus H$ is feasible for the correct $(R_+,H,(t_r))$: packing holds, each tiny-mass slab holds by the choice of $t_r$, and each covering holds because $\mu_r(S^*)\ge\Lambda_+$ on $R_+$.
Lemma~\ref{lem:lp-drop} returns an integral residual $R$ after dropping at most $F_{\max}$ fractionals.
Zeroing preserves packing, so $S:=H\cup R\in\F$.
Each covering coordinate of a residual item is at most $1$, so after the drop
\[
\mu_r(S)
=
\mu_r(H)+\mu_r(R)
\ge
\mu_r(H)+\bigl(\Lambda_+-\mu_r(H)\bigr)_+-F_{\max}
\ge
\Lambda
\qquad(r\in R_+).
\]
(If the covering right-hand side is zero then already $\mu_r(H)\ge\Lambda_+\ge\Lambda$.)
On every $r\notin R_+$ each dropped item contributes at most $\eta_{\mathrm{LP}}$, so
\[
\bigl\lvert\mu_r(R)-t_r\bigr\rvert
\le
\eps^2+F_{\max}\eta_{\mathrm{LP}}
\le
\tfrac32\eps^2,
\]
hence $\lvert\mu_r(R)-\mu_r(S^*\setminus H)\rvert\le\tfrac52\eps^2\le 4\eps^2$.
Every residual coordinate is at most $\eta_{\mathrm{LP}}$, so Lemma~\ref{lem:occ-robust} applies (via Le Cam, not via aligned $\ell_1$).
Thus $\lvert\E[\Topk(S)]-\E[\Topk(S^*)]\rvert\le O(\eps)\,W$.

After $\eps\leftarrow\Theta(\eps_{\mathrm{out}}/k)$ as in Corollary~\ref{cor:ptas}, the error is $O(\eps_{\mathrm{out}})\,\OPT$.
The $n$-exponent remains a function of $d$ and $1/\eps_{\mathrm{out}}$ only: throughout this case the internal parameter satisfies $k=O(1/\eps^2)$, and after the substitution one has $h=O(1/\eps_{\mathrm{out}}^6)$ and $h\Lambda_+(k\Lambda_++F_{\max})/\eps^2=O_d(1/\eps_{\mathrm{out}}^{O(1)})$.

The two cases meet at $k=\Theta(1/\eps^2)$.
\end{proof}

\section{Covering LP realization}
\label{app:covering}

\begin{proof}[Proof of Theorem~\ref{thm:covering}]
Run Algorithm~\ref{alg:min-transfer} with a $(1+O(\eps))$-approximate min-sum oracle given by the covering analogue of~\citet{frieze1984}, and set $W:=\E[\Topk(S_0)]$.
Theorem~\ref{thm:min-transfer} yields $\OPT_{\min}\le W\le O(\OPT_{\min})$.
Fold and grid at this $W$.
Throughout, $\eps$ is internal accuracy and $\eps_{\mathrm{out}}$ is the requested PTAS accuracy, as in the proof of Theorem~\ref{thm:packing}.
Let $S^*$ be a min-optimum of the preprocessed instance ($\E[\Topk(S^*)]\le W$).
The algorithm enumerates every guess below, realizes an integral covering from each feasible residual LP by Lemma~\ref{lem:lp-lift}, and returns the set of smallest folded $\E[\Topk]$.

\paragraph{Case $k\ge C/\eps^2$.}
Run the preprocessing of Algorithm~\ref{alg:quantile} at this $W$ and let $T^*$ be a mixture quantile of $S^*$.
Set $\eta:=\eps$ and call $i$ surplus-heavy if $\E[\Tr(X_i,T^*)]>\eta W$.
Lemma~\ref{lem:c1} and $\E[\Topk(S^*)]\le W$ give $\Phi_{S^*}(T^*)=O(W)$, hence $O(1/\eps)$ surplus-heavies.
Enumerate every such $H$ in time $n^{O(1/\eps)}$ and guess a residual target $z$ for the triple of $S^*\setminus H$.
The residual LP is~\eqref{eq:sig-lp-cov} on $[n]\setminus H$ with residual demands $(b-a(H))_+$, $D=3$, $C=0$, and the $\eps$-grid of Theorem~\ref{thm:quantile}.
Lemma~\ref{lem:lp-lift} raises at most $d+6$ fractionals.
Occupancy of the integral residual increases by $O(1)$, a window slack $\delta=O(1)$ in Lemma~\ref{lem:c1}; surplus increases by at most $(d+6)\eta W=O(\eps)\,W$.
Matched types therefore satisfy $\E[\Topk(H\cup S)]\le(1+O(\eps))\,\E[\Topk(S^*)]$ on the folded laws, as in the proof of Corollary~\ref{cor:min-allk}.

\paragraph{Case $k\le C/\eps^2$.}
Use $h$, $\Lambda$, $F_{\max}$, $\Lambda_+$, and $\eta_{\mathrm{LP}}$ as in the proof of Theorem~\ref{thm:packing}.
Guess $R_+$ and enumerate LP-heavies of $S^*$ in time $n^{O(h\Lambda_+(k\Lambda_++F_{\max})/\eps^2)}$.
The residual LP is~\eqref{eq:sig-lp-cov} with the same tiny-mass slabs and the same raised-layer covering inequalities as in that proof.
Lemma~\ref{lem:lp-lift} raises at most $F_{\max}$ fractionals.
Each non-raised residual mean increases by at most $F_{\max}\eta_{\mathrm{LP}}\le\eps^2/2$, so $\lvert\mu_r(R)-\mu_r(S^*\setminus H)\rvert\le 4\eps^2$ after the same quantization.
Raised layers only gain mass and remain at least $\Lambda$.
Lemma~\ref{lem:occ-robust} yields $\lvert\E[\Topk(S)]-\E[\Topk(S^*)]\rvert\le O(\eps)\,W$ on the folded laws.
After $\eps\leftarrow\Theta(\eps_{\mathrm{out}}/k)$ the error is $O(\eps_{\mathrm{out}})\,\OPT_{\min}$ and the $n$-exponent is $f(d,1/\eps_{\mathrm{out}})$.

In both cases, after the same internal-accuracy substitution as in the proof of Corollary~\ref{cor:min-allk}, the folded value $U$ of the returned set satisfies $U\le\OPT_{\min}+2\eps W\le(1+4\eps)W$ for $\eps\le 1/20$.
Lemma~\ref{lem:fold-back} therefore yields
\[
\E[\Topk(S;X)]
\le
\frac{U}{1-4\eps}
\le
(1+5\eps)\,U
\le
(1+O(\eps))\,\OPT_{\min}.
\]
\end{proof}

\section{Proofs for Section~\ref{sec:exactsum}}
\label{app:exactsum}
\label{app:runtime}

\begin{proof}[Proof of Theorem~\ref{thm:main}]
Let $S^*$ be optimal and let $S$ be a set with $\Sg(S)=\Sg(S^*)$ returned by the enumeration (exact-sum finds one whenever $S^*$ is feasible).
Theorems~\ref{thm:disc} and~\ref{thm:sig} with $c=2$ give
\[
\E[\Topk(S)]
\ge\E[\Topk(S^*)]-O((4+k)\eps)\,W
\ge\bigl(1-O((k+1)\eps)\bigr)\OPT.
\]
On an unsaturated layer the number of big items is at most $\Lambda/\eta=k\Lambda^2/\eps^2$, so each of the $B\le 2k\Lambda^2/\eps^4$ bucket counts is an integer in $\{0,\dots,k\Lambda^2/\eps^2\}$.
Saturated layers store one bit.
The tiny coordinates contribute $n^{O(h)}=n^{O(1/\eps^2)}$ encodings.
The histogram encodings contribute
\[
\bigl(k\Lambda^2/\eps^2\bigr)^{O(hB)}
=(k\Lambda/\eps)^{O(k\Lambda^2/\eps^6)}.
\]
The product is the stated running time, and each exact-sum instance has the same magnitude.
The proof of Corollary~\ref{cor:ptas} is the substitution $\eps\leftarrow\eps/(C(k+1))$ in these bounds: $h$ becomes $\Theta(k^2/\eps^2)$ and $hB$ becomes $O(k^7\Lambda'^2/\eps^6)$.
On $k\le C/\eps^2$ the exponent of $n$ is $O(1/\eps^6)$.
\end{proof}

\begin{proof}[Proof of Theorem~\ref{thm:quantile}]
Let $S^*$ be optimal for the original instance.
The fold changes $\E[\Topk]$ by $O(\eps)\,W$ (Lemma~\ref{lem:fold} with $c=2$), and the multiplicative rounding changes it by a further $O(\eps)\,\OPT$ (Lemma~\ref{lem:mult-round} with the same $c$).
Work with the folded, rounded laws, and write $S^*$ also for an optimal set of those laws.
If $\lvert S^*\rvert\le k$, then $\Topk(S^*)=\sum_{i\in S^*}X_i=A_0(S^*)$.
At $T=0$ one has $\Pr[X_i\ge 0]=1$, so the second coordinate of the triple determines $\lvert S\rvert$ exactly.
The triple of $S^*$ is among the enumerated targets, and any $S$ with the same triple satisfies $\lvert S\rvert=\lvert S^*\rvert\le k$ and $\lvert\Phi_S(0)-\Phi_{S^*}(0)\rvert\le\eps W$, hence $\E[\Topk(S)]=\E[A_0(S)]\ge\E[\Topk(S^*)]-O(\eps)\,W$.

Now assume $\lvert S^*\rvert\ge k$, and let $T^*$ be a mixture quantile of $S^*$ (Definition~\ref{def:mix-q}); after Lemma~\ref{lem:mult-round} one may take $T^*\in G\cup\{0\}$.
Exact-sum returns a set $S$ whose quantized triple at $T^*$ equals that of $S^*$.
Then $\lvert\mu_S^>(T^*)-\mu_{S^*}^>(T^*)\rvert\le O(\eps)$, $\lvert\mu_S^\ge(T^*)-\mu_{S^*}^\ge(T^*)\rvert\le O(\eps)$, and $\lvert\Phi_S(T^*)-\Phi_{S^*}(T^*)\rvert\le O(\eps)\,W$
(each coordinate of the triple is rounded to units of size $\Theta(\eps)$ after summing at most $n$ terms).
In particular $T^*$ is an $O(\eps)$-approximate mixture quantile of $S$, so Lemma~\ref{lem:c1} applies to both $S$ and $S^*$ with slack $\delta=O(\eps)$.
The hypothesis $k\ge C/\eps^2$ with $C\ge 2$ makes $\eta_k+O((\eps+1)/k)=O(\eps)$.
Therefore
\begin{align*}
\E[\Topk(S)]
&\ge
(1-O(\eps))\,\E[A_{T^*}(S)]\\
&\ge
(1-O(\eps))\bigl(\E[A_{T^*}(S^*)]-O(\eps)\,W\bigr)\\
&\ge
(1-O(\eps))\,\E[\Topk(S^*)]-O(\eps)\,W
\ge
(1-O(\eps))\,\E[\Topk(S^*)],
\end{align*}
using $\E[\Topk(S^*)]\ge W$.
Evaluating the true objective on every realized set can only help.
The running time is $O(\eps^{-1}\log(n/\eps))$ candidate thresholds, and $\poly(n,1/\eps)$ mixed-radix targets at each threshold, hence $\poly(n,1/\eps)$ exact-sum calls of magnitude $\poly(n,1/\eps)$.
(Enumerating every distinct input atom as a candidate $T$ also hits $T^*$ exactly, with no multiplicative rounding.)
\end{proof}

A signature stores, at each of $h=1/\eps^2$ layers, a tiny count in $\{0,\dots,n\cdot\eta/\gamma\}$ and a histogram over $B\le 1/\beta=2k\Lambda^2/\eps^4$ buckets, where $\Lambda=2k+C_0\log(k/\eps)$ and $\eta/\gamma=n/(k\Lambda)$.
Thus $hB=2k\Lambda^2/\eps^6$.
Unsaturated layers have at most $\Lambda/\eta=k\Lambda^2/\eps^2$ big items, so each bucket count lies in $\{0,\dots,k\Lambda^2/\eps^2\}$.
Saturated layers store a single bit.
Mixed radix therefore produces
\[
n^{O(1/\eps^2)}\,(k\Lambda/\eps)^{O(k\Lambda^2/\eps^6)}
\]
integer targets, each of bit length
\[
O\bigl((1/\eps^2)\log n+(k\Lambda^2/\eps^6)\log(k\Lambda/\eps)\bigr).
\]
Substituting $\eps\leftarrow\eps/(C(k+1))$ as in Corollary~\ref{cor:ptas} replaces the exponent of $n$ by $O(k^2/\eps^2)$.
On the range $k\le C/\eps^2$ of Corollary~\ref{cor:allk} that exponent is $O(1/\eps^6)$.
Theorem~\ref{thm:quantile} uses a different encoding, of magnitude $\poly(n,1/\eps)$, and does not pass through this substitution.

\section{Hardness reductions}
\label{app:hardness}

The hardness statements below are proved at $k=1$ (hence they apply to every $k\ge 1$).
A $\rho$-approximation of a maximization problem means $\mathrm{ALG}\ge\rho\cdot\OPT$ with $\rho\in(0,1]$.
The reductions go from deterministic max-sum to $\E[\max]$ on the same family: an algorithm for the stochastic problem is used as a black box.
Theorems~\ref{thm:transfer}(i) and~\ref{thm:noptas} are uniform (no single algorithm for every $\F$ in the class); the witness families also have a per-instance lower bound.
Theorem~\ref{thm:limits} is uniform on the packing class: two-dimensional knapsack is a packing family.
Theorem~\ref{thm:min-limits} is the FPTAS lower bound on two-dimensional covering knapsack; there is no EPTAS claim.

\begin{lemma}
\label{lem:lottery-sand}
Let $w_i\ge 0$ and $q\in(0,1)$, and let $X_i=B(w_i,q)$ be independent.
For every $S\subseteq[n]$ write $w(S)=\sum_{i\in S}w_i$ and $\E[\max(S)]=\E[\max_{i\in S}X_i]$ (value $0$ if $S=\emptyset$).
Then
\[
q(1-q)^{|S|-1}\,w(S)
\;\le\;
\E[\max(S)]
\;\le\;
q\,w(S),
\]
and therefore $q(1-q)^{n-1}w(S)\le\E[\max(S)]\le q\,w(S)$.
\end{lemma}

\begin{proof}
The upper bound uses $\max_{i\in S}X_i\le\sum_{i\in S}X_i$ and does not need independence.
For the lower bound let $A_i$ be the event that $X_i=w_i$ and $X_j=0$ for every $j\in S\setminus\{i\}$.
The events $A_i$ are disjoint, $\max=w_i$ on $A_i$, and $\Pr[A_i]=q(1-q)^{|S|-1}$ by independence, so
$\E[\max(S)]\ge\sum_{i\in S}w_i\Pr[A_i]=q(1-q)^{|S|-1}w(S)$.
If $S\neq\emptyset$ then $|S|-1\le n-1$ and $0<1-q<1$, hence $(1-q)^{|S|-1}\ge(1-q)^{n-1}$.
\end{proof}

\begin{lemma}
\label{lem:bern-ineq}
For $q\in[0,1]$ and an integer $n\ge 1$, $(1-q)^{n-1}\ge 1-(n-1)q$.
\end{lemma}

\begin{proof}
The bound $(1-q)^t\ge 1-tq$ holds for $t=n-1\ge 0$ by induction on $t$.
\end{proof}

Write $\OPT_w=\max_{S\in\F}w(S)$ and $\OPT_E=\max_{S\in\F}\E[\max(S)]$.
Lemma~\ref{lem:lottery-sand} yields $q(1-q)^{n-1}\OPT_w\le\OPT_E\le q\,\OPT_w$.

\begin{lemma}
\label{lem:lottery-transfer}
If an algorithm returns $S\in\F$ with $\E[\max(S)]\ge\rho\,\OPT_E$, then
$w(S)\ge\rho(1-q)^{n-1}\OPT_w$.
\end{lemma}

\begin{proof}
The upper bound of Lemma~\ref{lem:lottery-sand} gives $w(S)\ge\E[\max(S)]/q$.
Substitute $\E[\max(S)]\ge\rho\,\OPT_E$ and $\OPT_E\ge q(1-q)^{n-1}\OPT_w$.
\end{proof}

\begin{lemma}
\label{lem:lottery-transfer-min}
Write $\OPT_w^{\min}=\min_{S\in\F}w(S)$ and $\OPT_E^{\min}=\min_{S\in\F}\E[\max(S)]$.
Lemma~\ref{lem:lottery-sand} yields $q(1-q)^{n-1}\OPT_w^{\min}\le\OPT_E^{\min}\le q\,\OPT_w^{\min}$.
If an algorithm returns $S\in\F$ with $\E[\max(S)]\le\rho\,\OPT_E^{\min}$, then
$w(S)\le\rho\,(1-q)^{1-n}\,\OPT_w^{\min}$.
\end{lemma}

\begin{proof}
The lower bound of Lemma~\ref{lem:lottery-sand} gives $w(S)\le\E[\max(S)]/\bigl(q(1-q)^{n-1}\bigr)$.
Substitute $\E[\max(S)]\le\rho\,\OPT_E^{\min}$ and $\OPT_E^{\min}\le q\,\OPT_w^{\min}$.
\end{proof}

\begin{proof}[Proof of Theorem~\ref{thm:transfer}(i)]
Let $G=(V,E)$ be an undirected graph, $n=|V|$, and let $\F$ be the family of independent sets of $G$ (membership is polynomial).
Set $w_v=1$ and $X_v=B(1,q)$ with $q=1/n^2$ (the case $n=1$ is trivial).
Then $\OPT_w=\alpha(G)$.
\citet{zuckerman2006}: for every $\delta>0$, unless $P=NP$, maximum independent set has no $n^{1-\delta}$-approximation, and in particular no constant-factor approximation.

Suppose a single polynomial algorithm $\rho$-approximates $\E[\max]$ on every such Bernoulli instance, for a constant $\rho\in(0,1]$.
Lemmas~\ref{lem:lottery-transfer} and~\ref{lem:bern-ineq} give
\[
|S|
\ge
\rho(1-q)^{n-1}\alpha(G)
\ge
\rho\bigl(1-(n-1)/n^2\bigr)\alpha(G)
\ge
\rho(1-1/n)\alpha(G).
\]
For $n\ge 2$ this is a $(\rho/2)$-approximation of independent set, a contradiction.
(Hamiltonian path does not substitute: the gap $(n-1)$ versus $(n-2)$ is $1-O(1/n)$ and does not rule out a constant factor.)
\end{proof}

\begin{proof}[Proof of Theorem~\ref{thm:noptas}]
Let $G=(V,E)$ be an undirected graph, write $n=|E|$, and let $\F=\{\delta(U):U\subseteq V\}$ (membership is polynomial).
Set $w_e=1$ and $X_e=B(1,q)$ with $q=1/n^2$.
Then $\OPT_w$ is the unweighted max-cut value.
Max-cut is $\mathrm{APX}$-complete~\citep{papadimitriou1991}; unless $P=NP$ it has no polynomial approximation better than $16/17$~\citep{hastad2001}, and therefore no PTAS.

If a single PTAS for $\E[\max]$ existed on every cut family (in particular on this $\F$), the choice $\eps=1/100$ would return a cut with
\[
w(\delta(U))
\ge
0.99\,(1-q)^{n-1}\OPT_w
\ge
0.99(1-1/n)\OPT_w
\]
by Lemmas~\ref{lem:lottery-transfer} and~\ref{lem:bern-ineq}.
For $n\ge 20$ the ratio exceeds $16/17$, a contradiction.

Deterministic edge weights would make $\max_{e\in\delta(U)}X_e$ the heaviest edge of the cut, which does not encode cut weight.
The sparse lotteries are essential.
The same family still has a constant-factor algorithm by Theorem~\ref{thm:transfer}(ii) and Goemans--Williamson.
\end{proof}

\begin{proof}[Proof of Theorem~\ref{thm:limits}]
Let $\F$ be the feasible sets of a two-dimensional $0$-$1$ knapsack instance with profits $v_i$ and capacities $(A,B)$.
This $\F$ is a packing family with $d=2$, so Theorem~\ref{thm:packing} applies; the DP in time $O(nAB)$ is not query-weight when $A$ and $B$ are binary.
Unless $P=NP$, the deterministic profit problem has no FPTAS~\citep{magazine1984}.
Unless $W[1]=FPT$, it has no EPTAS~\citep{kulik2010}.
(The deterministic problem does have a PTAS~\citep{frieze1984}.)

Suppose $\E[\max]$ on this $\F$ has an FPTAS $\mathcal{A}$.
Given a knapsack instance $I$ and $\eps\in(0,1)$, form $I'$ by setting $X_i=B(v_i,q)$ with $q=\eps/(2n)$ (bits $O(\log n+\log(1/\eps))$) and run $\mathcal{A}(I',\eps/2)$.
Lemma~\ref{lem:bern-ineq} gives $(1-q)^{n-1}\ge 1-\eps/2$, and Lemma~\ref{lem:lottery-transfer} yields
\[
v(S)
\ge
\bigl(1-\eps/2\bigr)(1-q)^{n-1}\OPT_w
\ge
\bigl(1-\eps/2\bigr)^2\OPT_w
\ge
(1-\eps)\OPT_w.
\]
The running time is polynomial in $|I|$ and $1/\eps$, so two-dimensional knapsack would have an FPTAS.

If instead $\mathcal{A}$ is an EPTAS, the same instance $I'$ and the same call $\mathcal{A}(I',\eps/2)$ run in time $f(2/\eps)\,|I'|^{O(1)}$.
The extra $O(n\log(n/\eps))$ bits in $I'$ keep the time of EPTAS form, and two-dimensional knapsack would have an EPTAS.

The firing probability must depend on $\eps$.
A fixed $q=1/n^2$ leaves a $\Theta(1/n)$ relative gap, which swallows the target $\eps$ once $\eps\ll 1/n$.
The definition of an FPTAS or an EPTAS allows the reduced instance to depend on $\eps$.
Do not apply the Garey--Johnson integer-OPT criterion directly to $\E[\max]$: the objective is not integer.
One-dimensional knapsack and cardinality cannot be used here, because deterministic max-sum already has an FPTAS (and cardinality $\E[\max]$ has an EPTAS).
\end{proof}

\begin{lemma}
\label{lem:cover-nofptas}
Unless $P=NP$, deterministic min-sum on two-dimensional $0$-$1$ covering knapsack admits no FPTAS, already when every cost equals $1$.
\end{lemma}

\begin{proof}
The source problem is Subset Sum~\citep{garey1979}: given nonnegative integers $y_1,\ldots,y_n$ and a target $T$, all encoded in binary, decide whether some $S\subseteq[n]$ has $\sum_{i\in S}y_i=T$.
If $T=0$ the answer is yes.
If $T>\sum_i y_i$ the answer is no.
Assume $1\le T\le\sum_i y_i$ and write $W:=1+\sum_i y_i$.
For each cardinality $m\in\{1,\ldots,n\}$ form a covering instance $I^{(m)}$ with $n$ items of unit cost, covering coefficients $(y_i+W,\,W-y_i)$, and demands $(T+mW,\,mW-T)$.
The second demand is positive because $T<W$.

Let $S\subseteq[n]$ have size $s$.
The two covering inequalities are $y(S)+sW\ge T+mW$ and $sW-y(S)\ge mW-T$, i.e.
\[
y(S)\ge T+(m-s)W
\qquad\text{and}\qquad
y(S)\le T+(s-m)W.
\]
If $s<m$ then $T+(m-s)W\ge T+W>W-1\ge y(S)$, so $S$ is infeasible.
If $s=m$ the two bounds collapse to $y(S)=T$.
If $s>m$ both inequalities hold for every $S$, since $0\le y(S)\le W-1$.
Hence $I^{(m)}$ is feasible whenever $n>m$, and when $n=m$ if and only if $\sum_i y_i=T$.
Its optimum is at most $m$ if and only if some $m$-set sums to $T$.

Subset Sum is a yes-instance if and only if some $I^{(m)}$ has optimum at most $m$.
On a feasible $I^{(m)}$ the optimum is an integer in $\{1,\ldots,n\}$.
An FPTAS run at $\eps=1/(2n+2)$ therefore returns a value strictly less than $\mathrm{OPT}+1$ and decides each $I^{(m)}$ exactly, in time polynomial in the bit length of $(y,T)$.
Both inequalities are covering inequalities on the original items: the construction does not complement a packing feasible set.
\end{proof}

\begin{proof}[Proof of Theorem~\ref{thm:min-limits}]
Let $\F$ be the feasible sets of a two-dimensional $0$-$1$ covering knapsack instance with costs $v_i$ and demands $(A,B)$.
This $\F$ is a covering family with $d=2$, so Theorem~\ref{thm:covering} applies.
Deterministic min-sum on this $\F$ has a PTAS by the covering form of~\citet{frieze1984}, and has no FPTAS unless $P=NP$ by Lemma~\ref{lem:cover-nofptas}.
The paper of~\citet{kulik2010} is not used.

Suppose $\min\E[\max]$ on this $\F$ has an FPTAS $\mathcal{A}$.
Given a covering instance $I$ and $\eps\in(0,1)$, form $I'$ by setting $X_i=B(v_i,q)$ with $q=\eps/(2n)$ and run $\mathcal{A}(I',\eps/2)$.
Lemma~\ref{lem:bern-ineq} gives $(1-q)^{n-1}\ge 1-\eps/2$, and Lemma~\ref{lem:lottery-transfer-min} with $\rho=1+\eps/2$ yields
\[
v(S)
\le
\frac{1+\eps/2}{(1-q)^{n-1}}\,\OPT_w^{\min}
\le
\frac{1+\eps/2}{1-\eps/2}\,\OPT_w^{\min}
\le
(1+O(\eps))\,\OPT_w^{\min}.
\]
The running time is polynomial in $|I|$ and $1/\eps$, so two-dimensional covering knapsack would have an FPTAS, contradicting Lemma~\ref{lem:cover-nofptas}.
One-dimensional covering knapsack cannot be used: deterministic min-sum already has an FPTAS.
\end{proof}

\end{document}